\documentclass[twocolumn]{autart}    
\usepackage{graphicx}     
\usepackage{amsfonts,latexsym,
amssymb,
amsmath,url,stackengine}
\usepackage{booktabs}
\usepackage{amsmath}
\usepackage{picins}
\usepackage{color}

\catcode`\@=11
\def\downparenfill{$\m@th\braceld\leaders\vrule\hfill\bracerd$}
\def\downparenfill{$\m@th\braceld\leaders\vrule\hfill\bracerd$}
\def\overparen#1{\mskip 2mu\mathop{\vbox{\ialign{##\crcr\crcr \noalign{\kern0.4ex}
\downparenfill\crcr\noalign{\kern0.4ex\nointerlineskip}
$\hfil\displaystyle{#1}\hfil$\crcr}}}\limits\mskip 2mu} 
\catcode`\@=12

\usepackage{rgsMacros}

\usepackage{subcaption}
\usepackage{bm}
\let\labelindent\relax
\usepackage{enumitem}
\usepackage{balance}
\usepackage{comment}

\definecolor{olivegreen}{rgb}{0,0.5,0}

\newtheorem{thmm}{Theorem}
\newtheorem{lemm}{Lemma}
\newtheorem{exe}{Example}
\newtheorem{corol}{Corollary}
\newtheorem{ass}{Assumption}
\newtheorem{prope}{Property}
\newtheorem{defin}{Definition}
\newtheorem{remm}{Remark}

\newenvironment{lemma}{\begin{lemm}}{\hfill \end{lemm}}

\newenvironment{remark}{\begin{remm} \rm}{\hfill \end{remm}}

\newenvironment{assumption}{\begin{ass}}{\end{ass}}

\newenvironment{theorem}{\begin{thmm}}{\hfill \end{thmm}}

\newenvironment{proof}{{\it Proof. }}{\hfill $\blacksquare$ }

\newtheorem{dwellt}{Condition}

\newtheorem{feedback}{Feedback}

\newtheorem{adapt_alg}{Algorithm}
\newenvironment{adapt}{\begin{adapt_alg}}{\hfill \end{adapt_alg}}

\newif\ifitsdraft
\def\itsdraft{\global\itsdrafttrue}

\itsdraft  
 
\begin{document}

\begin{frontmatter}
\title{Adaptive Boundary Control of the Kuramoto-Sivashinsky Equation Under Intermittent Sensing} 

\author{M. C. Belhadjoudja}\ead{mohamed.belhadjoudja@gipsa-lab.fr}, 
\author{M. Maghenem}\ead{mohamed.maghenem@gipsa-lab.fr},
\author{E. Witrant}\ead{emmanuel.witrant@gipsa-lab.fr},
\author{C. Prieur}\ead{christophe.prieur@gipsa-lab.fr} 

\address{The authors are with the Universit\'e Grenoble Alpes, CNRS, Grenoble-INP, GIPSA-lab, F-38000, Grenoble, France}
       

\begin{abstract}
We study in this paper boundary stabilization, in the $L^2$ sense, of the perturbed Kuramoto-Sivashinsky (KS) equation subject to intermittent sensing. We assume that we measure the state on a given spatial subdomain during certain time intervals, while we measure the state on the remaining spatial subdomain during the remaining time intervals. We assign a feedback law at the boundary of the spatial domain and force to zero the value of the state at the junction of the two subdomains. Throughout the study, the equation's destabilizing coefficient is assumed to be unknown and possibly space dependent but bounded. As a result, adaptive boundary controllers are designed under different assumptions on the perturbation. In particular, we guarantee input-to-state stability (ISS) when an upperbound on the perturbation’s size is known. Otherwise,
only global uniform ultimate boundedness (GUUB) is guaranteed. In contrast, when the state is measured at every spatial
point all the time (full state measurement), convergence to an arbitrarily-small neighborhood of the origin is guaranteed, even
if the perturbation’s maximal size is unknown. Numerical simulations are performed to illustrate our results.
\end{abstract}

\end{frontmatter}

\section{Introduction} \label{Introduction}

In control loops, inputs are often subject to sensor limitations. For systems governed by partial differential equations (PDEs) evolving in both space and time, a major limitation is the spatial range of sensing, as it can be impractical to measure the state simultaneously at every spatial point. 
Another limitation is the energy cost associated with always-on sensors, as it can be unaffordable to run the sensor  all the time. In this work, we propose an intermittent-sensing scenario that takes into account the aforementioned limitations, in which the system’s state is measured at certain spatial  subdomains over specific time intervals, rather than continuously across the entire spatial domain.
This scenario depicts situations of restricted-energy sensors \cite{intermittent_app1,intermittent_app5}, network control systems \cite{intermittent_app4,ew}, and mobile sensors \cite{intermittent_app3}. 

The proposed sensing scenario is considered in the context of the perturbed KS equation given by \cite{kuramoto78,sivashinsky80}
\begin{align}
u_{t}+uu_{x}+ \lambda(x) u_{xx}+u_{xxxx} = f(x,t), \label{KS}
\end{align}
where $x \in (0,1)$ is the 
one-dimensional space variable, $\lambda >0$ is known as the \textit{destabilizing coefficient}, and $f$ is the perturbation. This equation has been used to model various physical phenomena including wildfires \cite{wildfire}, turbulence in reaction-diffusion systems \cite{kuramoto78}, and plasma instabilities \cite{plasmaInstable1}, among others. 

Since \cite{liu2001stability}, several boundary controllers are proposed to stabilize the origin for \eqref{KS} when $f\equiv 0$;  see \cite{kobayashi,coron2015fredholm,HugoKS}. These works, however, either assume a small $\lambda$, in which case boundary measurements are enough to ensure global asymptotic/exponential stability, or assume arbitrary $\lambda$ but small initial conditions while requiring full state measurement. Additionally, $\lambda$ is often assumed constant, though it could be space-dependent according to \cite{vs1,vs2}. To the best of our knowledge, the first work to study \eqref{KS} without constraining the range of the initial condition nor the size of $\lambda$ is \cite{KS1}. The state in the latter reference is assumed to be measured intermittently over the subdomains $(0,Y)$ and $(Y,1)$. Furthermore, feedback controllers are applied at $x=0$ and $x=1$, while enforcing a zero state at $x=Y$, to ensure $L^2$ global exponential stability of the origin. However, the destabilizing coefficient $\lambda$ is assumed constant and known. This requirement is relaxed in \cite{ACC23-KS} by allowing $\lambda$ to be unknown, and proposing an adaptive
design of the control parameters, to guarantee the same stability conclusions. On the other hand, the perturbed KS equation is considered in \cite{caoInh}, where a boundary controller using boundary measurements is shown to guarantee an input-to-state stability  with respect to $f$, under small $\lambda$.

In this work, we study boundary stabilization for \eqref{KS} under intermittent sensing, in the presence of the bounded space- and time-dependent perturbation $f$. More specifically, we extend the adaptive Lyapunov-based approach in \cite{ACC23-KS} to the perturbed case, when an upper bound on the norm of $f$ is either known or unknown. In the former case, we establish $L^2$-input-to-state stability \cite{ISS_PDE} with respect to $f$, while in the latter case, we establish only $L^2$-global uniform ultimate boundedness. On the other hand, under full state measurement, we simplify our design to achieve convergence to arbitrarily-small neighborhoods of the origin, even if we ignore an upper bound on the norm of $f$. Notably, 
the presence of perturbations prevents us from establishing the same properties under full and intermittent sensing, different from the unperturbed case, where $L^2$ global exponential stability of the origin is guaranteed under both scenarios.

The rest of the paper is organized as follows. Problem formulation is in Section \ref{Problem formulation}. The proposed adaptive 
 controller as well as key Lyapunov inequalities are in Section \ref{preliminaries}. The results are in Sections \ref{results} and \ref{GpA}. Finally, numerical simulations are in Section \ref{numerical}. 

\ifitsdraft 

\else 

Due to space limitations, certain proofs and intermediate lemmas have been omitted and can be found in \cite{preprint}.

\fi 

\textbf{Notation.}  
Depending on the context, a.e. means either almost every or almost everywhere. We denote by $L^2(a,b)$, for $b>a$, the space of functions $u:[a,b]\to \mathbb{R}$ such that $\int_{a}^{b}u(x)^2dx<+\infty$. Furthermore, we let $|u|_{\infty} := \esssup_{x\in (a,b)}|u(x)|:=\inf \{M\geq 0\ : \ |u(x)|\leq M \ \ \text{for a.e.} \ x\in (a,b)\}$. For $(x,t) \mapsto u(x,t)$, the partial derivative of $u$ with respect to $t$ is denoted by $u_t$, the first partial derivative with respect to $x$ is denoted by $u_x$, the second partial derivative with respect to $x$ is denoted by $u_{xx}$ (and so on), and we may write $u(x)$ instead of $u(x,t)$. We denote the time derivative of a function $t \mapsto V(t)$ by $\dot{V}$. We also denote the space derivative of a function $x \mapsto \lambda(x)$ of a scalar variable by $\lambda'$. A continuous function $\Phi : \mathbb{R}_{\geq 0}\to \mathbb{R}_{\geq 0}$ is of class $\mathcal{K}$ if it vanishes at zero and is strictly increasing. Finally, for $ x \in \mathbb{R}$,  $\sign (x) = 1$ if $x>0$, $=0$ if $x=0$ and $=-1$ if $x<0$.

\section{Intermittent sensing and control location} \label{Problem formulation}

In this section, we formulate the proposed sensing scenario and the in-domain and boundary conditions.

\subsection{Intermittent sensing} \label{Intermittent sensing}
Consider equation \eqref{KS}, let $Y \in (0,1)$, and consider a sequence 
$\{t_i\}^{\infty}_{i=1} \subset \mathbb{R}_{\geq 0}$, where $t_1 = 0$ and $t_{i+1} > t_i$ for all $i\in \{1,2,...\}$, such that
\begin{enumerate}[label={
S\arabic*)},leftmargin=*]
\item \label{s1} $u(x,t)$ is measured for all $t \in I_{1}:= \bigcup^{\infty}_{k = 1} [t_{2k-1},t_{2k})$ and for a.e. $x\in (0,Y)$.
\item \label{s2} $u(x,t)$ is measured for all $t\in I_{2}:= \bigcup^{\infty}_{k = 1} [t_{2k},t_{2k+1})$ and for a.e. $x\in (Y,1)$.
\end{enumerate}
Associated with the proposed sensing scenario, we consider the following assumption.
\begin{assumption} \label{dwelltime}
There exist four constants $\overline{T}_1$, $\overline{T}_{2}$, $\underline{T}_{1}$, $\underline{T}_{2}>0$ such that, for each $k \in \mathbb{N}^{*}$, we have
\begin{align*}
\underline{T}_{1}  \leq t_{2k} - t_{2k-1} \leq \overline{T}_1, 
\quad 
\underline{T}_{2}  \leq t_{2k+1} - t_{2k} \leq \overline{T}_2.
\end{align*}
\end{assumption}

The proposed sensing scenario could represent situations where two battery-powered sensors are used (one measuring $u$ over $(0,Y)$, and the other one over $(Y,1)$), in  case where their simultaneous activation can be costly; see \cite{intermittent_app1,intermittent_app5}. Furthermore, in the context of network control systems \cite{network}, we can assume that the two sensors share the same channel to transfer measurements to the controller. Hence, each sensor is allowed to use the channel only over certain time intervals. Another motivation emerges when using a single mobile sensor \cite{scanning}. Indeed, when the sensor’s commuting speed (between two subdomains $(0,Y)$ and $(Y,1)$) is fast enough, the proposed sensing scenario can offer a fair approximation.

\subsection{Control locations} \label{Intermittent control}

In the case of intermittent sensing, we propose to control \eqref{KS} at $x=0$ and $x=1$. Also, we set the $u$ and its spatial derivative $u_x$ at $x=Y$ to a null value. Therefore, we assimilate \eqref{KS} to a system of two PDEs interconnected through boundary constraints at $x=Y$. That is, we introduce the system of PDEs
\begin{subequations}
\label{twopdesB}
\begin{equation}
\label{twopdes}
\begin{aligned} 
w_{t}  + w w_{x} + \lambda w_{xx}+w_{xxxx} & = f \qquad x\in (0,Y),
\\
v_{t}  + v v_{x} + \lambda v_{xx}+v_{xxxx} & = f \qquad x\in (Y,1),
\end{aligned}
\end{equation}
\begin{equation} \label{1X1BC}
\begin{aligned}
w(Y)& = w_{x}(Y) = w_{x}(0) = 0, 
\\
v(Y) & = v_{x}(Y) = v_{x}(1) = 0, \\
w(0) & = u_{1}, ~ v(1) = u_{2},
\end{aligned}
\end{equation}
\end{subequations}
where $(u_{1},u_{2})$ are control  inputs.

As a consequence, we define a solution $u : [0,1] \times \mathbb{R}_{\geq 0} \rightarrow \mathbb{R}$ to  \eqref{KS} subject to the boundary conditions
\begin{equation}
\label{eqBKS}
\begin{aligned}
u(Y) & = u_{x}(0) = u_{x}(Y) = u_{x}(1)= 0, \\
u(0) & = u_{1},~ u(1) = u_{2}, 
\end{aligned}
\end{equation}
for some inputs $(u_{1},u_{2})$,  as $u(x,t) := w(x,t)$ for a.e. $(x,t)\in (0,Y) \times \mathbb{R}_{> 0}$ and $u(x,t) := v(x,t)$ for a.e. $(x,t)\in (Y,1) \times \mathbb{R}_{> 0}$, with $(w,v)$ a strong solution to \eqref{twopdesB} subject to the same inputs $(u_1,u_2)$; see \cite[Chapter 9]{strong} for the concept of strong solutions. According to this concept of solutions, we study \eqref{KS} under \eqref{eqBKS} by fully focusing on \eqref{twopdesB}.

\section{General approach}  \label{preliminaries}

We follow in this work a Lyapunov-based approach to design $(u_1, u_2)$ for \eqref{twopdesB}. To do so, we start introducing the Lyapunov function candidates
\begin{align} \label{Vs}
\hspace{-0.4cm} V_{1}(w) := \frac{1}{2}\int_{0}^{Y}w(x)^2dx, ~~ V_{2}(v) := \frac{1}{2}\int_{Y}^{1}v(x)^2dx. 
\end{align}

The following lemma establishes key Lyapunov inequalities along \eqref{twopdesB}, that hold under the following assumption.     
\begin{assumption} \label{asssupf}
The function $\lambda$ is absolutely continuous on $(0,Y)\cup (Y,1)$, and there exist $\bar{f}$, $\bar{\lambda}_l$, $\bar{\lambda}_r$, $\bar{\lambda}'_l$, $\bar{\lambda}'_r > 0$ such that $|f|_\infty \leq \bar{f}$ and
\begin{align*}
 \esssup_{x\in (0,Y)}|\lambda (x)| & \leq \bar{\lambda}_l, \  \ \esssup_{x\in (Y,1)}|\lambda (x)| \leq \bar{\lambda}_r,
\\
\esssup_{x\in (0,Y)} |\lambda'(x)| & \leq \bar{\lambda}'_l, \  \  \esssup_{x\in (Y,1)} |\lambda'(x)| \leq \bar{\lambda}'_r.
\end{align*}
\end{assumption}
 
\begin{lemma}\label{v_space_vary}
Along \eqref{twopdesB}, it holds that 
\begin{align}
&\dot{V}_1 \leq \theta_{1}V_1+C_1\sqrt{V_1}+ \frac{u_1^3}{3}+u_1w_{xxx}(0), \label{v1_new} \\
&\dot{V}_2 \leq \theta_2V_2+C_2\sqrt{V_2}-\frac{u_2^3}{3}-u_2v_{xxx}(1), \label{v2_new}
\end{align}
where $C_{1} := \sqrt{2Y}\bar{f}, ~ C_{2} := \sqrt{2(1-Y)}\bar{f}$, and 
\begin{align*}
\theta_1 & := 
{\bar{\lambda}_{l}}^{'2}
+ 2 \left(\bar{ \lambda}_{l}+\frac{1}{2}\right)\left(\left(
\bar{\lambda}_{l} + 
\frac{1}{2} \right) + 
\frac{12}{Y^2}\right), 
 \\ 
\theta_2 & := 
{\bar{\lambda}_{r}}^{'2} 
+ 2\left({\bar{\lambda}}_{r} 
+
\frac{1}{2}\right)  \left(\left(\bar{\lambda}_{r} +
\frac{1}{2}\right)+\frac{12}{(1-Y)^2}\right).
\end{align*}
\end{lemma}
\begin{proof}
The proof is in \ifitsdraft the Appendix \else \cite{preprint} \fi.
\end{proof}

In view of \eqref{v1_new}-\eqref{v2_new}, the control input $u_1$ affects the (worst-case) behavior of $V_1$ only and since $(V_1,w_{xxx}(0))$ is available on $I_1$ only, we propose to design $u_1$ to guarantee a stable behavior for $V_1$ over $I_1$ and set $u_1 = 0$ otherwise. For the same reasons, we design $u_2$ to guarantee a stable behavior for $V_2$ over $I_2$ and set $u_2 = 0$ otherwise. That is,  we let 
\begin{equation}
\label{eqInputs}
\begin{aligned} 
  (u_{1}, u_{2}) :=
\left\{
\begin{matrix} 
(\kappa (V_{1},w_{xxx}(0),\hat{\theta}_{1}),0) &  \quad \text{on} ~ I_{1}, 
\\
(0,-\kappa (V_{2},v_{xxx}(1),\hat{\theta}_{2})) & \quad \text{on}~  I_{2}, 
\end{matrix} \right.
\end{aligned}
\end{equation}
where $\hat{\theta}_1, \hat{\theta}_2 \geq 0$ are adaptation parameters 
that we design to handle the fact that $\lambda$, $\lambda'$, and $(\overline{T}_{1},\overline{T}_{2},\underline{T}_{1}, \underline{T}_{2})$ 
from Assumption \ref{dwelltime} can be unknown.
Furthermore, we design $\kappa : \mathbb{R}_{\geq 0}\times \mathbb{R}\times \mathbb{R}_{\geq 0}\to \mathbb{R}$ to verify 
\begin{align}
\frac{\kappa(V,\omega,\hat{\theta})^3}{3}+\kappa(V,\omega,\hat{\theta})\omega \leq - \hat{\theta}V. \label{kappa}
\end{align} 
The latter inequality combined to  \eqref{v1_new}-\eqref{v2_new} yields 
\begin{equation} \label{eq.switch}
\begin{aligned}
\left\{
\begin{matrix}
\begin{matrix}
\dot{V}_1 \leq & (\theta_{1}-\hat{\theta}_{1})V_{1} + C_1\sqrt{V_1}
\\
\dot{V}_2 \leq & \theta_{2} V_2 + C_2\sqrt{V_2}
\end{matrix}
\qquad \text{a.e. on}~ I_{1},
\\   ~~ \\ 
\begin{matrix}
\dot{V}_1 \leq & \theta_{1} V_1+C_1\sqrt{V_1}
\\
\dot{V}_2 \leq & (\theta_{2}-\hat{\theta}_{2})V_{2}+C_2\sqrt{V_2}
\end{matrix}
\qquad \text{a.e. on}~ I_{2}.
\end{matrix}
\right.
\end{aligned}
\end{equation}
As a result, the stability analysis of the origin for the system of  PDEs \eqref{twopdesB} is achieved through the stability analysis of the origin for the switched system \eqref{eq.switch}. 

We provide in the next lemma an example of a feedback law $\kappa$ verifying \eqref{kappa}.

\begin{lemma} 
\label{lem_control}
Given $\varepsilon,\delta >0$   such that $\delta \varepsilon \geq 1$, $\varepsilon (\delta \varepsilon) \geq 1$, and $\varepsilon^3\delta^2-(\varepsilon/3)-\varepsilon \delta - 1 \geq 0$, the function $\kappa : \mathbb{R}_{\geq 0} \times \mathbb{R} \times \mathbb{R}_{\geq 0} \rightarrow \mathbb{R}$ given by 
\begin{equation}\label{eq:kappa}
    \kappa(V,\omega,\hat{\theta}) := 
    \left\lbrace 
    \begin{aligned}
 -\sign(\omega)   \sqrt[3]{V} \qquad  & \text{if $|\omega| \geq l(V,\hat{\theta})$},\\
  - \varepsilon (\hat{\theta}  + \delta ) \sqrt[3]{V}  \qquad  &\text{otherwise,}
    \end{aligned}
    \right. 
\end{equation}
where $l(V,\hat{\theta}) := (1/3) (1 + 3 \hat{\theta}) V^{2/3}$,  
verifies \eqref{kappa}.
\end{lemma}
\begin{proof}
The proof is in \ifitsdraft the Appendix \else \cite{preprint} \fi.
\end{proof}

\begin{remark}\label{chattering}
The discontinuity of $\kappa$ is key to verify \eqref{kappa} while guaranteeing boundedness of $\kappa$ independently of how large $\omega$ can be, and as long as $V$ and $\hat{\theta}$ are bounded. This will allow us to conclude that, as long as $(V_1,V_2,\hat{\theta}_1,\hat{\theta}_2)$ are bounded, the control inputs $(u_1,u_2)$ will remain bounded independently of how large $(w_{xxx}(0),v_{xxx}(1))$ may be. Another discontinuity in $(u_1,u_2)$ is due to transitioning from a time interval in $I_{1}$ to a time interval in $I_{2}$, and vice versa; see \eqref{eqInputs}. Discontinuous boundary control of PDEs brings challenges in terms of numerical and practical implementation as well as to ensure well-posedness; see \cite{discontinuous1,discontinuous2}.
\end{remark}

\section{Results under intermittent sensing} \label{results}

In this section, we design the parameters $(\hat{\theta}_1,\hat{\theta}_2)$ in \eqref{eqInputs} when $\bar{f}$ is unknown. Then, to ensure stronger properties, we assume that $\bar{f}$ is known.    

\subsection{Unknown perturbation range} \label{GUUB}

In this case, we update $(\hat{\theta}_1,\hat{\theta}_2)$ according to the following algorithm.
\begin{adapt}\label{adapt_algo_2}
Given $\Delta_{1}, \Delta_{2}$, $\sigma >0$, the coefficients $(\hat{\theta}_{1},\hat{\theta}_{2})$ are dynamically updated as follows.
\begin{enumerate}[label={
R\arabic*)},leftmargin=*]
\item \label{item:T41} 
 On every interval $[t_{2k-1},t_{2k}) \subset I_1$, we set $\dot{\hat{\theta}}_2:=0$. Moreover, if 
 \begin{align*}
\hspace{-0.7cm} & V_{1}(t_{2k-1})  > V_{1}(t_{2k-3})\exp^{-\sigma (t_{2k-1}-t_{2k-3})} \nonumber \\
\hspace{-0.7cm} & +\left(\hat{\theta}_{1}(t_{2k-3})+\frac{\hat{\theta}_{1}(t_{2k-3})^{2}}{4}\right)  \exp^{\left(\hat{\theta}_{1}(t_{2k-3})+1\right) \left( t_{2k-1}-t_{2k-3} \right)}, 
\end{align*}
 we set $\dot{\hat{\theta}}_{1}:= \Delta_{1}$; otherwise, we set $\dot{\hat{\theta}}_{1}:=0$.
\item \label{item:T42}  
On every interval $[t_{2k},t_{2k+1}) \subset I_2$, we set $\dot{\hat{\theta}}_{1}:=0$. Moreover, if 
 \begin{align*}
\hspace{-0.7cm} & V_{2}(t_{2k})  > V_{2}(t_{2k-2})\exp^{-\sigma (t_{2k}-t_{2k-2})} \nonumber \\
 \hspace{-0.7cm} & +\left(\hat{\theta}_{2}(t_{2k-2})+\frac{\hat{\theta}_{2}(t_{2k-2})^{2}}{4}\right)  \exp^{\left( \hat{\theta}_{2}(t_{2k-2})+1 \right) \left(t_{2k}-t_{2k-2}\right)},
 \end{align*}
 we set $\dot{\hat{\theta}}_{2}:= \Delta_{2}$; otherwise, we set $\dot{\hat{\theta}}_{2}:=0$.
\item \label{item:T43}  On $[t_{1},t_{3}]$, $\hat{\theta}_{1}=\hat{\theta}_{1}(0)\geq 0$ and $\hat{\theta}_{2}=\hat{\theta}_{2}(0) \geq 0$. 
\end{enumerate}
\end{adapt}

\begin{remark}\label{rem_theta}
According to Algorithm \ref{adapt_algo_2}, if we assume that $t\mapsto \hat{\theta}_1(t)$ is bounded (which we prove later), being non-decreasing, it will become equal to a constant $\theta_1^*$ after some finite time. Hence, according to \ref{item:T41}, we verify that $k\mapsto V_1(t_{2k-1})$ converges to the ball centered at the origin of radius $\left(\theta_1^*+\frac{\theta_1^{*2}}{4}\right)  \exp^{(\theta_1^*+1) \left( \bar{T}_1+\bar{T}_2 \right)}$. As we will show that $\theta_1^*$ admits an upperbound that is independent of the PDEs initial conditions, we will be able to conclude that $V_1$ converges to a neighborhood of the origin whose size is independent of the PDEs initial condition. The same reasoning applies to $(\hat{\theta}_2,V_2)$.
\end{remark}

\begin{remark}
The tuning parameters $\Delta_1$ and $\Delta_2$ represent the adaptation gains for $\hat{\theta}_1$ and $\hat{\theta}_2$, respectively. Specifically, $\Delta_1$ (respectively, $\Delta_2$) determines the rate at which $t \mapsto \hat{\theta}_1(t)$ (respectively, $t \mapsto \hat{\theta}_2(t)$) increases during intervals in $I_1$ (respectively, $I_2$) until specific decay of $t \mapsto V_1(t)$ (respectively, $t \mapsto V_2(t)$) is detected, which marks the end of the adaptation phase. Hence, larger values of $\Delta_1$ (respectively, $\Delta_2$) will result in $\hat{\theta}_1$ (respectively, $\hat{\theta}_2$) quickly compensating the destabilizing effect of $\theta_1$ (respectively, $\theta_2$). 
It is to be noted that, since the growth rate of $t \mapsto \hat{\theta}_1(t)$ (respectively, $t \mapsto \hat{\theta}_2(t)$) is updated at specific time instants, large values for $\Delta_1$ (respectively, $\Delta_2$) can lead to large values for $\hat{\theta}_1$ (respectively, $\hat{\theta}_2$) and thus large control inputs. On the other hand, the gain $\sigma$ defines the desired decay rate for $t \mapsto (V_1(t) + V_2(t))$, once the adaptation phase is over. That is, larger values for $\sigma$ would result in faster decay rates.
\end{remark}

\begin{theorem}\label{theorem2}
Consider system \eqref{twopdesB} such that Assumption \ref{asssupf} holds. 
Consider the intermittent sensing scenario in \ref{s1}-\ref{s2} in Section \ref{Intermittent sensing} such that Assumption \ref{dwelltime} holds.  Let $(V_{1},V_2)$ be defined in \eqref{Vs}, $\kappa$ be defined in \eqref{eq:kappa}, and $(\hat{\theta}_{1},\hat{\theta}_{2})$ be governed by Algorithm \ref{adapt_algo_2}. Then,  system \eqref{twopdesB} in closed loop with $(u_{1},u_{2})$ as in \eqref{eqInputs} satisfies:
\begin{itemize}
    \item $L^2$-global uniform ultimate boundedness: For any  $(\hat{\theta}_{1}(0),\hat{\theta}_{2}(0))$, there exists a constant $r>0$ such that for any $R>0$, there exists a finite time $T(R)\geq 0$ such that, along the closed-loop solutions,  we have 
\begin{center}
$
\hspace{-0.3cm} V_{1}(0)+V_{2}(0)\leq R \Rightarrow V_{1}(t)+V_{2}(t)\leq r, ~ \forall t\geq T(R).
$ \end{center}
    \item For each $(\hat{\theta}_{1}(0),\hat{\theta}_{2}(0))$, there exists $M>0$ such that $\max 
    \{|\hat{\theta}_{1}|_{\infty}, |\hat{\theta}_{2}|_{\infty} \} \leq M$ for all $(V_{1}(0),V_{2}(0))$.
    \item The inputs $(u_1,u_2)$ remain bounded.
\end{itemize}
\end{theorem} 

\begin{proof}
The proof relies on the stability analysis of the origin $\{(V_1,V_2)=0\}$ for the switched system \eqref{eq.switch} and is divided into three steps. In Step $1$, we prove that $(\hat{\theta}_{1},\hat{\theta}_{2})$ become constant after some finite time $T\geq 0$. Furthermore, $\hat{\theta}_{1}$ and $\hat{\theta}_{2})$ never exceed $\max 
\{\hat{\theta}_1(0), M_1\}$ and $\max 
\{\hat{\theta}_2(0),M_2\}$, respectively, where $M_1, M_2>0$ depend only on $(\sigma, \theta_{1},\theta_{2}, C_1, C_2, \overline{T}_{1},\overline{T}_{2},\underline{T}_{1},\underline{T}_{2},\Delta_{1}, \Delta_{2})$. 
In Step $2$, we analyze the function $V_1+V_2$ and show that the $L^2$-global uniform ultimate boundedness property holds. Finally, in Step $3$, we use the structure of $\kappa$ and the fact that $(\hat{\theta}_1,\hat{\theta}_2, V_1,V_2)$ are bounded to conclude that $(u_1,u_2)$ remains bounded. 

\textit{Step 1:} We first show that $\hat{\theta}_{1}$ and $\hat{\theta}_{2}$ become constant after some $T \geq 0$ using contradiction. Namely, we assume that there is not such a $T \geq 0$ such that, for all $t\geq T$,  $\dot{\hat{\theta}}_{1}(t)=\dot{\hat{\theta}}_{2}(t)=0$.  As a result, according to \ref{item:T41}-\ref{item:T43} in Algorithm \ref{adapt_algo_2}, we conclude that 
$\lim_{t \to \infty}\hat{\theta}_{1}(t)=\lim_{t\to \infty}\hat{\theta}_{2}(t) = \infty$. Therefore, there exists $T \geq 0$ such that, for all $t\geq T$, we have
\begin{align}
&\hat{\theta}_1(t) \geq \theta_1+C_1+\frac{(\theta_1+1)\bar{T}_2+\sigma (\bar{T}_1+\bar{T}_2)}{\underline{T}_1}+1, \label{15} \\
&\hat{\theta}_2(t) \geq \theta_2+C_2+\frac{(\theta_2+1)\bar{T}_1+\sigma (\bar{T}_1+\bar{T}_2)}{\underline{T}_2}+1. \label{152}  
\end{align}
Let $k>1$ be such that $t_{2k-3}\geq T$. Using \ifitsdraft Lemma \ref{to_use_proofs1} in the Appendix, \else Lemma 4 in \cite{preprint}, \fi while replacing therein $(V,\hat{\theta},\theta,C)$ by $(V_1,\hat{\theta}_1,\theta_1,C_1)$, we find
\begin{equation}
\label{17proof}
\begin{aligned}
 V_{1}(t_{2k-2}) & \leq V_{1}(t_{2k-3})\exp^{-\sigma (t_{2k-2}-t_{2k-3})} 
\\ & + \hat{\theta}_{1}(t_{2k-3}). 
\end{aligned}
\end{equation}
Next, we shall prove the following inequality
\begin{equation}
\label{22_to_prove} 
\begin{aligned}
V_{1}&(t_{2k-1})  \leq V_{1}(t_{2k-2})\exp^{(\theta_{1}+1)(t_{2k-1}-t_{2k-2})}  \\
&~+\frac{\hat{\theta}_{1}(t_{2k-3})^2}{4}\exp^{(\hat{\theta}_{1}(t_{2k-3})+1)(t_{2k-1}-t_{2k-2})}. 
\end{aligned}
\end{equation}
To obtain the latter inequality, we start using
$$ \dot{V}_1 \leq \theta_1 V_1 + C_1 \sqrt{V_1} \quad \text{a.e. on } [t_{2k-2},t_{2k-1}) \subset I_1. $$
Then, applying \ifitsdraft Lemma \ref{ineqsquare} in the Appendix \else Lemma 5 in \cite{preprint} \fi while replacing $(\theta,C,[0,T],\delta)$ therein by $(\theta_1,C_1,[t_{2k-2},t_{2k-1}], 1)$, we obtain 
\begin{align}
 V_1(t_{2k-1}) & \leq V_1(t_{2k-2})\exp^{(\theta_1+1)(t_{2k-1}-t_{2k-2})} \nonumber\\
&+\frac{C_1^2}{4(\theta_1+1)}\exp^{(\theta_1+1)(t_{2k-1}-t_{2k-2})} \nonumber \\
&\leq \left( V_1(t_{2k-2}) +\frac{C_1^2}{4} \right) \exp^{(\theta_1+1)(t_{2k-1}-t_{2k-2})}. \nonumber 
\end{align}
Using the fact that $\hat{\theta}_1(t_{2k-3}) \geq \theta_1$, we obtain
\begin{equation}
\label{to_use_later2}
\begin{aligned}
V_1(t_{2k-1}) &\leq V_1(t_{2k-2})\exp^{(\theta_1+1)(t_{2k-1}-t_{2k-2})} \\
&+\frac{C_1^2}{4}\exp^{(\hat{\theta}_1(t_{2k-3})+1)(t_{2k-1}-t_{2k-2})}. 
\end{aligned}
\end{equation}
Finally, using \eqref{to_use_later2} and the fact that $\hat{\theta}_1(t_{2k-3})\geq C_1$, we obtain \eqref{22_to_prove}.

By combining \eqref{17proof} and \eqref{22_to_prove}, we obtain
\begin{equation}
\label{later}
\begin{aligned}
&V_{1}(t_{2k-1})\leq V_{1}(t_{2k-3})\exp^{-\sigma (t_{2k-1}-t_{2k-3})} + \\
&\left(\hat{\theta}_{1}(t_{2k-3})+\frac{\hat{\theta}_{1}(t_{2k-3})^{2}}{4}\right) \exp^{(\hat{\theta}_{1}(t_{2k-3})+1)\left(t_{2k-1}-t_{2k-3}\right)}, 
\end{aligned}
\end{equation}
which implies, according to \ref{item:T41} in Algorithm \ref{adapt_algo_2}, that $\dot{\hat{\theta}}_{1}(t)=0$ for all $t\geq t_{2k-1}$. We show in a similar way that $\dot{\hat{\theta}}_{2}(t)=0$ for all $t\geq t_{2k}$, which leads to a contradiction. 

To show the uniform boundedness of $\hat{\theta}_1$ with respect to the closed-loop trajectories, we first suppose the existence of $T\geq 0$ such that 
\begin{equation}\label{theta1_eq}
    \hat{\theta}_{1}(T) = \theta_{1}+C_{1}+\frac{(\theta_{1}+1)\overline{T}_{2}+\sigma (\overline{T}_{1}+\overline{T}_{2})}{\underline{T}_{1}} + 1.
\end{equation}
Since $\hat{\theta}_1$ is non-decreasing, then either $\hat{\theta}_1$ is smaller than the right-hand side of \eqref{theta1_eq}, or $\hat{\theta}_1(0)$ is greater than the right-hand side of \eqref{theta1_eq}, or there exists $T$ such that \eqref{theta1_eq} holds. In the case where such a $T$ exists, we can always pick it to be in $I_{1}$, as $\hat{\theta}_1$ is constant over each interval in $I_2$. Thus, we let $T \in [t_{2k'-3},t_{2k'-2}) \subset I_{1}$ for some $k'>1$. Note that inequality \eqref{theta1_eq} implies \eqref{15}, with $t\geq T$. We have already shown that if \eqref{15} holds, then \eqref{later} also holds for all $k>1$ such that $t_{2k-3}\geq T$. Thus, we can write that
\begin{align*}
&V_{1}(t_{2k+1})\leq V_{1}(t_{2k-1})\exp^{-\sigma (t_{2k+1}-t_{2k-1})} + \nonumber \\
&\left(\hat{\theta}_{1}(t_{2k-1})+\frac{\hat{\theta}_{1}(t_{2k-1})^{2}}{4}\right) \exp^{(\hat{\theta}_{1}(t_{2k-1})+1)\left(t_{2k+1}-t_{2k-1}\right)},
\end{align*}
for all $k>1$ such that $t_{2k-1}\geq T$. Furthermore, using \ref{item:T31} in Algorithm \ref{adapt_algo_2}, we conclude that 
$ 0 \leq \dot{\hat{\theta}}(t)\leq \Delta_1$ for all $t\geq 0$.   
As a result, 
$$\hat{\theta}_1(T) \leq \hat{\theta}_1(t_{2k'+1})\leq \hat{\theta}_1(T) + 2\Delta_1 \overline{T}_1. $$
On the other hand, note that, using \ref{item:T31}, we have $\dot{\hat{\theta}}_1(t)=0$ for all $t \geq t_{2k'+1}$. As a result, 
\begin{align*}
\hat{\theta}_1(t) \leq \hat{\theta}_1(T)+2\Delta_1 \overline{T}_1, \quad \forall t\geq t_{2k'+1}. 
\end{align*}
We thus conclude that $|\hat{\theta}_{1}|\leq \max \{\hat{\theta}_1(0),M_1\}$, with 
\begin{align}
M_1 := \theta_{1}+C_{1}+\frac{(\theta_{1}+1)\overline{T}_{2}+\sigma (\overline{T}_{1}+\overline{T}_{2})}{\underline{T}_{1}}  +1+2\Delta_{1}\overline{T}_{1}. \label{M1}
\end{align}
Similarly, we can show that $|\hat{\theta}_{1}|\leq \max \{\hat{\theta}_1(0),M_1\}$, with 
\begin{align}
M_{2} := \theta_{2}+C_{2} + \frac{(\theta_{2}+1)\overline{T}_{1} + \sigma (\overline{T}_{1}+\overline{T}_{2})}{\underline{T}_{2}} + 1 + 2\Delta_{2}\overline{T}_{2}. \label{M2}
\end{align}

\textit{Step 2:} To study the function $V_{1}+V_{2}$, we first define the sequences  $\{T_{i}\}_{i=0}^{\infty}$ and $\{T_{i}'\}_{i=1}^{\infty}$, such that $T_{i}:=t_{2i+1}$ and $T_{i}':=t_{2i}$. Since $\hat{\theta}_1$ and $\hat{\theta}_2$ are nondecreasing and become constant after some finite time, we conclude the existence of at most a finite number of intervals $[t_{2k-1},t_{2k}) \subset I_1$, on which, $\hat{\theta}_1$ may increase. On the latter intervals, we know that $V_1$ is governed by the inequality $\dot{V}_1 \leq \theta_1V_1+C_1\sqrt{V_1}$, while on the remaining intervals, $V_1$ does not verify the inequality in \ref{item:T31}. The same reasoning applies to $\hat{\theta}_2$ and $V_2$. More precisely, for each initial condition $(\hat{\theta}_{1}(0),\hat{\theta}_{2}(0))$, there exist two integers $N^{*}_{1}, N_{2}^{*} \in \mathbb{N}$ such that, for each locally absolutely continuous solution $(V_{1},V_{2})$ to \eqref{eq.switch}, there exist two finite increasing subsequences
$\{i_1,i_2,...,i_{N_1^*}\}
\subset \mathbb{N}$ and 
$
\{j_1,j_2,...,j_{N_2^*}\} \subset \mathbb{N}^*$ such that
\begin{itemize}
    \item For each $i\in \{i_1,i_2,...,i_{N_{1}^{*}}\}$, we have
    \\
$
   V_{1}(T_{i+1})  \leq \left(V_{1}(T_{i})+\frac{M_{1}^2}{4}\right)\exp^{(M_1 + 1)(T_{i+1}-T_{i})}. 
    $
    \item For each $j\in \{j_1,j_2,...,j_{N_{2}^{*}}\}$, we have 
   \\
  $ V_{2}(T_{j+1}')  \leq \left(V_{2}(T_{j}')+\frac{M_{2}^2}{4}\right)\exp^{(M_2 + 1)(T_{j+1}'-T_{j}')}$. 
    \item For each $i\in \mathbb{N}/\{i_1,i_2,...,i_{N_{1}^{*}}\}$, we have
    \\
    $V_{1}(T_{i+1}) \leq V_{1}(T_{i})\exp^{-\sigma (T_{i+1}-T_{i})} 
    \\
    +\left(M_1+\frac{M_{1}^2}{4}\right) \exp^{(M_1+1)(T_{i+1}-T_{i})}.$
    \item For each $j\in \mathbb{N}^{*}/\{j_1,j_2,...,j_{N_{2}^{*}}\}$, we have
    \\
    $V_{2}(T_{j+1}') \leq V_{2}(T_{j}')\exp^{-\sigma (T_{j+1}'-T_{j}')}\\
    +\left(M_2+\frac{M_{2}^2}{4}\right)\exp^{(M_2+1)(T_{j+1}'-T_{j}')}.$
\end{itemize}
  Using Lemma \ifitsdraft \ref{last_lem} in the Appendix\else 6 in \cite{preprint}\fi, while replacing  $(V,M,\psi,N^*)$ therein by $(V_1,M_1,M_1+1,N^*_1)$, we obtain 
\begin{align*}
V_{1}(T_{i}) & \leq \gamma_{1} V_{1}(0)\exp^{-\sigma T_i} + \Phi_{1}(M_1) \quad \forall i\in \mathbb{N},
\end{align*}
for some $\gamma_1 > 0$ and $\Phi_1 \in \mathcal{K}$. 
Similarly, using Lemma \ifitsdraft \ref{last_lem} in the Appendix\else 6 in \cite{preprint}\fi, while replacing  $(V,M,\psi,N^*, \{T_i\}^\infty_{i=1})$ therein by $(V_2,M_2,M_2+1,N^*_2, 
\{T'_i\}^\infty_{i=1})$, we obtain
\begin{align*}
V_{2}(T_{i}') & \leq \gamma_{2} V_{2}(0) \exp^{-\sigma T_{i}'} + \Phi_{2}(M_2) \quad \forall i\in \mathbb{N}^*,
\end{align*}
for some $\gamma_2 > 0$ and  $\Phi_2 \in \mathcal{K}$.   As a consequence, for each $t\in [T_i,T_{i+1}]$, 
\begin{align}
V_{1}(t) & \leq \gamma_1 \exp^{(\theta_1+1+\sigma)(\overline{T}_1+\overline{T}_2)}V_1(0)\exp^{-\sigma t} \nonumber \\
& + \left[\Phi_{1}(M_1) + \frac{M_{1}^2}{4(\theta_1+1)}\right]\exp^{(\theta_1+1)(\overline{T}_{1}+\overline{T}_{2})}. \label{final_ineq1_new}
\end{align}
Moreover, for each $t\in [T_{i}',T_{i+1}']$,  
\begin{align}
V_{2}(t) & \leq \gamma_2 \exp^{(\theta_2+2+\sigma)(\overline{T}_1+\overline{T}_2)}V_2(0)\exp^{-\sigma t} \nonumber \\
& + \left[\Phi_{2}(M_2) + \frac{M_{2}^2}{4(\theta_2+1)}\right]\exp^{(\theta_2+1)(\overline{T}_{1}+\overline{T}_{2})}. \label{final_ineq2_new}
\end{align}
Defining
\begin{align*}
\gamma & := \max \left\{\gamma_1 \exp^{(\theta_1+1+\sigma)(\overline{T}_1+\overline{T}_2)}, \gamma_2 \exp^{(\theta_2+2+\sigma)(\overline{T}_1+\overline{T}_2)}
\right\}, 
\\
\Phi & := \left[\Phi_{1}(M_1) + \frac{M_{1}^2}{4(\theta_1+1)}\right]\exp^{(\theta_1+1)(\overline{T}_{1}+\overline{T}_{2})} \nonumber \\
& + \left[\Phi_{2}(M_2) + \frac{M_{2}^2}{4(\theta_2+1)}\right]\exp^{(\theta_2+1)(\overline{T}_{1}+\overline{T}_{2})},
\end{align*}
and summing \eqref{final_ineq1_new} and \eqref{final_ineq2_new}, which are valid for all $t\geq 0$, we obtain 
\begin{align*}
V_{1}(t) + V_2(t) \leq \gamma (V_{1}(0)+V_2(0))\exp^{-\sigma t} + \Phi.  
\end{align*}
Let $r := \Phi + \epsilon$, where $\epsilon$ is any positive constant, and suppose that $V_{1}(0)+V_2(0)\leq R$. We conclude that $V_{1}(t) + V_2(t) \leq \gamma R\exp^{-\sigma t} + \Phi$ for all $t \geq 0$. Hence, to guarantee that $\gamma R\exp^{-\sigma t} + \Phi \leq r$,  it is sufficient to have $t \geq T(R) := \frac{1}{\sigma} \log\left(\frac{\gamma R}{\epsilon}\right)$. 

\textit{Step $3$:} The boundedness of $(u_1,u_2)$ follows from the boundedness of $(\hat{\theta}_1,\hat{\theta}_2)$, the $L^2$-global uniform ultimate boundedness property of the closed-loop solutions, and the structure of $\kappa$ in \eqref{eq:kappa} which implies that $|u_i|\leq \max\{\varepsilon (\hat{\theta}_i+\delta),1\}\sqrt[3]{V_i}$ for $i\in \{1,2\}$.
\end{proof}

\subsection{Known perturbation range} \label{ISS}

In this section, we use the knowledge of $(C_1,C_2)$ in \eqref{eq.switch} to design $(\hat{\theta}_1,\hat{\theta}_2)$ according to the following algorithm. 
\begin{adapt}\label{adapt_algo_1}
Given $\Delta_{1}$, $\Delta_{2}$, $\sigma >0$, we update $(\hat{\theta}_{1},\hat{\theta}_{2})$ according to the following rules.
\begin{enumerate}[label={
R\arabic*)},leftmargin=*]
\item \label{item:T31} 
 On each interval $[t_{2k-1},t_{2k}) \subset I_1$, we set $\dot{\hat{\theta}}_2:=0$. Moreover, if 
 \begin{align*}
 V_{1}(&t_{2k-1})  > V_{1}(t_{2k-3})\exp^{-\sigma (t_{2k-1}-t_{2k-3})} \nonumber \\
 & +\left(C_{1}+\frac{C_{1}^{2}}{4}\right) \exp^{\left(\hat{\theta}_{1}(t_{2k-3}) +1\right) \left(t_{2k-1}-t_{2k-3} \right)}, 
 \end{align*}
 we set $\dot{\hat{\theta}}_{1}:= \Delta_{1}$; otherwise, we set $\dot{\hat{\theta}}_{1}:=0$.
\item \label{item:T32}  
On each interval $[t_{2k},t_{2k+1}) \subset I_2$, we set $\dot{\hat{\theta}}_{1}:=0$. Moreover, if 
 \begin{align*}
 V_{2}(&t_{2k}) > V_{2}(t_{2k-2})\exp^{-\sigma (t_{2k}-t_{2k-2})} \nonumber \\
 &~+\left(C_2+\frac{C_{2}^{2}}{4}\right)\exp^{\left( \hat{\theta}_{2}(t_{2k-2})+1 \right) \left(t_{2k}-t_{2k-2}\right)},
 \end{align*}
 we set $\dot{\hat{\theta}}_{2}:= \Delta_{2}$; otherwise, we set $\dot{\hat{\theta}}_{2}:=0$.
\item \label{item:T33}  On $[t_{1},t_{3}]$, $\hat{\theta}_{1}=\hat{\theta}_{1}(0)\geq 0$ and $\hat{\theta}_{2}=\hat{\theta}_{2}(0)\geq 0$. 
\end{enumerate}
\end{adapt}

The main difference between Algorithms \ref{adapt_algo_2} and \ref{adapt_algo_1}, is that in the latter case, the inequality in \ref{item:T31} uses the constant $C_1$. As a result, along the lines of Remark \ref{rem_theta}, we show that the map $k\mapsto V_1(t_{2k-1})$ converges to the ball centered at the origin of radius $\left(C_1+\frac{C_1^2}{4}\right)  \exp^{(\theta_1^*+1) \left( \bar{T}_1+\bar{T}_2 \right)}$. Hence, since $C_1 := \sqrt{2Y}\bar{f}$, we conclude that the latter radius is a class $\mathcal{K}$ function of $\bar{f}$. The same reasoning applies for the behavior of $V_2$. 

\begin{theorem}\label{ISSthm}
Consider system \eqref{twopdesB} such that Assumption \ref{asssupf} holds. Consider the intermittent sensing scenario in \ref{s1}-\ref{s2} in Section \ref{Intermittent sensing} such that Assumption \ref{dwelltime} holds. 
Let $(V_{1},V_2)$ be defined in \eqref{Vs}, $\kappa$ be defined in \eqref{eq:kappa}, and the parameters $(\hat{\theta}_{1},\hat{\theta}_{2})$ governed by Algorithm \ref{adapt_algo_1}. Then, system \eqref{twopdesB} in closed loop with $(u_{1},u_{2})$ as in \eqref{eqInputs} satisfies:
\begin{itemize}
\item $L^2$-input-to-state stability with respect to $\bar{f}$; namely,  for each $(\hat{\theta}_{1}(0),\hat{\theta}_{2}(0))$, there exists $\gamma\geq 1$ and a class $\mathcal{K}$ function $\Phi$ such that, for all $t\geq 0$, we have 
\begin{align*}
V_{1}(t) + V_{2}(t) \leq \gamma (V_{1}(0)+V_{2}(0))e^{-\sigma t} + \Phi (\bar{f}),
\end{align*}    
where $\sigma>0$ comes from Algorithm \ref{adapt_algo_1}.
\item For each $(\hat{\theta}_{1}(0),\hat{\theta}_{2}(0))$, there exists $M>0$ such that $|\hat{\theta}_{1}|_{\infty}, |\hat{\theta}_{2}|_{\infty}
\leq M$  for all $(V_{1}(0),V_{2}(0))$.
\item The inputs $(u_{1},u_{2})$ remain bounded. Additionally, if $\bar{f}=0$, then $(u_1,u_2)$ converge to zero. 
\end{itemize}
\end{theorem}
\begin{proof}
The proof is divided into three steps. In Step 1, we prove that $\hat{\theta}_{1}$ and $\hat{\theta}_{2}$ become constant after some $T>0$ and are bounded uniformly with respect to $(V_1,V_2)$. In Step 2, we analyze the function $V_1+V_2$ and show that the $L^2$-input-to-state stability property with respect to $f$ is verified. Finally, the boundedness of $(u_1,u_2)$ follows, as in Step $3$ in the proof of Theorem \ref{thm_p}, from the structure of $\kappa$ and the boundedness of $(V_1,V_2,\hat{\theta}_1,\hat{\theta}_2)$.

\textit{Step 1:} We show that $\hat{\theta}_{1}$ and $\hat{\theta}_{2}$ are constant after some $T \geq 0$ using contradiction. Assume that there is not such a finite time $T \geq 0$ such that $ \dot{\hat{\theta}}_{1}(t)=\dot{\hat{\theta}}_{2}(t)=0$ for all $t\geq T$. As a result, according to \ref{item:T41}-\ref{item:T43} in Algorithm \ref{adapt_algo_1}, $\lim_{t\to \infty}\hat{\theta}_{1}(t)=\lim_{t\to \infty}\hat{\theta}_{2}(t) = \infty$. 
Therefore, there exists $T \geq 0$ such that, for all $t\geq T$, inequalities \eqref{15} and \eqref{152} hold. Let $k>1$ be such that $t_{2k-3}\geq T$. Inequality \eqref{to_use_later2} hold. Furthermore, using Lemma \ifitsdraft \ref{to_use_proofs1} in the Appendix \else 4 in \cite{preprint} \fi while replacing therein $(V,\hat{\theta},\theta,C)$ by $(V_1,\hat{\theta}_1,\theta_1,C_1)$, we obtain,
\begin{align}
V_1(t_{2k-2}) \leq V_1(t_{2k-3})\exp^{-\sigma (t_{2k-2}-t_{2k-3})} + C_1. \label{to_use_now} 
\end{align}
By combining \eqref{to_use_later2} and \eqref{to_use_now}, we obtain 
\begin{align*}
V_{1}(t_{2k-1}) & \leq V_{1}(t_{2k-3})\exp^{-\sigma (t_{2k-1}-t_{2k-3})} + \nonumber \\
&\left(C_1+ C_1^{2}/4 \right) \exp^{(\hat{\theta}_{1}(t_{2k-3})+1)\left(t_{2k-1}-t_{2k-3}\right)},
\end{align*}
which implies, according to \ref{item:T41} in Algorithm \ref{adapt_algo_1}, that $\dot{\hat{\theta}}_{1}(t)=0$ for all $t\geq t_{2k-1}$. We show in a similar way that $\dot{\hat{\theta}}_{2}(t)=0$ for all $t\geq t_{2k}$, which leads to a contradiction. We follow the exact same steps as in the proof of Theorem \ref{theorem2} to conclude that $|\hat{\theta}_1|\leq \max 
\{\hat{\theta}_1(0), M_1\}$ and $|\hat{\theta}_2|\leq \max 
\{\hat{\theta}_2(0),M_2\}$, where $M_1$ and $M_2$ are defined in \eqref{M1} and \eqref{M2} respectively.

\textit{Step 2:} To study the function $V_{1}+V_{2}$, we introduce the sequences  $\{T_{i}\}_{i=0}^{\infty}$ and $\{T_{i}'\}_{i=1}^{\infty}$, such that $T_{i}:=t_{2i+1}$ and $T_{i}':=t_{2i}$. As in the proof of Theorem \ref{theorem2}, for each initial conditions $(\hat{\theta}_{1}(0),\hat{\theta}_{2}(0))$, there exist two integers $N^{*}_{1}, N_{2}^{*}\in \mathbb{N}$ such that, for each locally absolutely continuous solution $(V_{1},V_{2})$ to \eqref{eq.switch}, there exist two finite increasing subsequences
$\{i_1,i_2,...,i_{N_1^*}\}\subset \mathbb{N}$ and $\{j_1,j_2,...,j_{N_2^*}\}\subset \mathbb{N}^*$ such that
\begin{itemize}
    \item For each $i\in \{i_1,i_2,...,i_{N_1^*}\}$, we have \\
    $V_{1}(T_{i+1}) \leq \left(V_{1}(T_{i})+\frac{C_{1}^2}{4}\right)\exp^{(M_1 + 1)(T_{i+1}-T_{i})}$.
    \item For each $j\in \{j_1,j_2,...,j_{N_2^*}\}$, we have \\ 
 $ V_{2}(T_{j+1}') \leq \left(V_{2}(T_{j}')+\frac{C_{2}^2}{4}\right)\exp^{(M_2 + 1)(T_{j+1}'-T_{j}')}$.
    \item For each $i\in \mathbb{N}/\{i_1,i_2,...,i_{N_1^*}\}$, we have \\
    $V_{1}(T_{i+1}) \leq V_{1}(T_{i})\exp^{-\sigma (T_{i+1}-T_{i})}\\
    +\left(C_1+\frac{C_{1}^2}{4}\right) \exp^{(M_1+1)(T_{i+1}-T_{i})}$.
    \item For each $j\in \mathbb{N}^{*}/\{j_1,j_2,...,j_{N_2^*}\}$, we have \\
    $V_{2}(T_{j+1}') \leq V_{2}(T_{j}')\exp^{-\sigma (T_{j+1}'-T_{i}')} \\
    +\left(C_2+\frac{C_{2}^2}{4}\right) \exp^{(M_2+1)(T_{j+1}'-T_{j}')}$.
\end{itemize}
 Using Lemma \ifitsdraft \ref{last_lem} in the Appendix\else 6 in \cite{preprint}\fi, while replacing  $(V,M,\psi,N^*)$ therein by $(V_1,C_1,M_1+1,N^*_1)$, we obtain the inequality
\begin{align}
V_{1}(T_{i}) & \leq \gamma_{1} V_{1}(0)\exp^{-\sigma T_i} + \Phi_{1}(C_1) \quad \forall i\in \mathbb{N}, \label{to_p1}
\end{align}
for some $\gamma_1 > 0$ and $\Phi_1 \in \mathcal{K}$. 
Similarly, using Lemma \ifitsdraft \ref{last_lem} in the Appendix\else 6 in \cite{preprint}\fi, while replacing  $(V,M,\psi,N^*, \{T_i\}^\infty_{i=1})$ therein by $(V_2,C_2,M_2+1,N^*_2, 
\{T'_i\}^\infty_{i=1})$, we obtain
\begin{align}
V_{2}(T_{i}') & \leq \gamma_{2} V_{2}(0) \exp^{-\sigma T_{i}'} + \Phi_{2}(C_2) \quad \forall i\in \mathbb{N}^*, \label{to_p2}
\end{align}
for some $\gamma_2 > 0$ and $\Phi_2 \in \mathcal{K}$.
As a consequence, for each $t\in [T_i,T_{i+1}]$, we have
\begin{align}
V_{1}(t) & \leq \gamma_1 \exp^{(\theta_1+1+\sigma)(\overline{T}_1+\overline{T}_2)}V_1(0)\exp^{-\sigma t} \nonumber \\
& + \left[\Phi_{1}(C_1) + \frac{C_{1}^2}{4(\theta_1+1)}\right]\exp^{(\theta_1+1)(\overline{T}_{1}+\overline{T}_{2})}. \label{final_ineq1}
\end{align}
Similarly, for each $t\in [T_{i}',T_{i+1}']$, we have 
\begin{align}
V_{2}(t) & \leq \gamma_2 \exp^{(\theta_2+2+\sigma)(\overline{T}_1+\overline{T}_2)}V_2(0)\exp^{-\sigma t} \nonumber \\
& + \left[\Phi_{2}(C_2) + \frac{C_{2}^2}{4(\theta_2+1)}\right]\exp^{(\theta_2+1)(\overline{T}_{1}+\overline{T}_{2})}. \label{final_ineq2}
\end{align}
Defining 
\begin{align*}
& \gamma  := 
\\ & 
\max \left\{\gamma_1 \exp^{(\theta_1+1+\sigma)(\overline{T}_1+\overline{T}_2)}, \gamma_2 \exp^{(\theta_2+2+\sigma)(\overline{T}_1+\overline{T}_2)}
 \right\},
\\
& \Phi(\bar{f}) :=\left[\Phi_{1}(C_1) + \frac{C_{1}^2}{4(\theta_1+1)}\right]\exp^{(\theta_1+1)(\overline{T}_{1}+\overline{T}_{2})} \nonumber \\
&~+ \left[\Phi_{2}(C_2) + \frac{C_{2}^2}{4(\theta_2+1)}\right]\exp^{(\theta_2+1)(\overline{T}_{1}+\overline{T}_{2})}, 
\end{align*}
and combining \eqref{final_ineq1} and \eqref{final_ineq2}, we obtain, for all $t\geq 0$, 
\begin{align}
\hspace{-0.3cm} V_{1}(t) + V_2(t) \leq \gamma (V_{1}(0)+V_2(0))\exp^{-\sigma t} + \Phi (\bar{f}). 
\end{align}
\end{proof}

\begin{remark}
The main result in \cite{ACC23-KS} follows now as a direct corollary of Theorem \ref{ISSthm}. Indeed, using the latter theorem with $f\equiv 0$, we recover the $L^2$ global exponential stability of the origin.
\end{remark}

\section{Result under full sensing} \label{GpA}

In this section,
we consider \eqref{KS} subject to  
\begin{align}
u(0) = u_1, \ \ u_x(0)=u(1)=u_x(1)=0,\label{bc}
\end{align}
where $u_1$ is a control input to be designed. We assume that $\lambda$ is absolutely continuous on $(0,1)$ and that $u$ is measured a.e. on $(0,1)$ and for all time. As a result, defining the Lyapunov function candidate 
\begin{align}
V(u) := \frac{1}{2}\int_{0}^{1}u(x)^2dx, \label{V_new}
\end{align}
and using Lemma \ref{v_space_vary}, we conclude that, along \eqref{KS} and \eqref{bc},   
\begin{align}
\dot{V} \leq \theta V + C \sqrt{V} + \frac{u_1^3}{3}+u_1u_{xxx}(0), ~ C := \sqrt{2}\bar{f} \label{36}
\end{align}
where 
$ \theta := |\lambda'|_{\infty}^{2}+2\left(|\lambda|_{\infty}+\frac{1}{2}\right)\left(\left(|\lambda|_{\infty}+\frac{1}{2}\right)+12\right).$

Now, by letting $u_1 := \kappa (V,u_{xxx}(0),\hat{\theta})$, where $\kappa$ is defined in \eqref{eq:kappa} and $\hat{\theta} > 0$ to be designed, we conclude using \eqref{36} and Lemma \ref{lem_control} that
\begin{align*}
\dot{V} \leq (\theta - \hat{\theta})V + C \sqrt{V}. 
\end{align*}
Finally, using Young inequality, we obtain
\begin{align} \label{eqYoungafter}
\dot{V} \leq \left(\theta + \frac{C^2}{\epsilon}-\hat{\theta}\right)V+\epsilon \quad \forall \epsilon>0.   
\end{align}
At this point, we propose to update the parameter $\hat{\theta}$ according to the following algorithm.

\begin{adapt}\label{adapt_0}
Given $\Delta$, $\tau$, $\epsilon$, $\sigma >0$,  the coefficient $\hat{\theta}$ is dynamically updated, on each interval $[k\tau,(k+1)\tau]$ with $k \in \mathbb{N}^{*}$, according to the following rules: 
\begin{enumerate}[label={
R\arabic*)},leftmargin=*]
\item \label{item:RR1} 
For each $t \in [k\tau, (k+1) \tau]$ if 
 \\
 $$ V(s) \leq V(k\tau) \exp^{-\sigma (s-k\tau)}+\frac{\epsilon}{\sigma} \quad  \forall s \in [k\tau, t],$$ 
then $\dot{\hat{\theta}}(t)=0$; otherwise, 
\begin{align} \label{eqjump}
\hat{\theta}(r)=\hat{\theta}(k\tau)+\Delta \quad \forall r\in [t,(k+1)\tau]. 
\end{align}

\item \label{item:RR4} On the interval $[0,\tau)$, we set $\hat{\theta} = \hat{\theta}(0)\geq 0$. 
\end{enumerate}
\end{adapt}

According to Algorithm \ref{adapt_0}, if $t\mapsto \hat{\theta}(t)$ is bounded, being non-decreasing, we can show that $\hat{\theta}$ becomes constant after some finite time. Hence, according to \ref{item:RR1} in Algorithm \ref{adapt_0}, the function $t\mapsto V(t)$ converges to a ball of radius $\epsilon /\sigma$. Since such a radius can be made arbitrarily small, we conclude that we can achieve the convergence of $V$ to an arbitrarily small neighborhood of the origin.

\begin{theorem}\label{thm_p}
Consider system \eqref{KS} under \eqref{bc} such that Assumption \ref{asssupf} holds. Let $V$ be defined in \eqref{V_new}, $\kappa$ be defined in \eqref{eq:kappa}, and the parameter $\hat{\theta}$ governed by Algorithm \ref{adapt_0}. Then, we conclude that \eqref{KS} and \eqref{bc} in closed-loop with $u := \kappa(V,u_{xxx}(0),\hat{\theta})$ satisfies the following properties: 
\begin{itemize}
\item  $L^2$ globally practically attractivity (GpA) of the origin; namely, for any $\eta>0$, we can find adaptation gains $(\tau, \sigma, \epsilon)$ such that, for any $\hat{\theta}(0) \geq 0$ and for any $V(0) \geq 0$, we have $\limsup_{t\to +\infty} V(t) 
\leq \eta$.
\item For each $\hat{\theta}(0)$, there exists $M>0$ such that $|\hat{\theta}|_{\infty} 
\leq M$ for all $V(0)$.
\item The input $u_{1}$ remains bounded.
\end{itemize}
\end{theorem}

\begin{proof}
We first show that $\hat{\theta}$ admits an upperbound that does not depend on $V(0)$. As a result, since $\hat{\theta}$ is nondecreasing and, when it increases, it does so according to \eqref{eqjump}, we conclude that $\hat{\theta}$ becomes constant after some finite time $T \geq 0$. Next, we analyze the Lyapunov function candidate $V$ in \eqref{V_new} and show that,  after a finite time, $V$ starts decaying exponentially towards a neighborhood of the origin,  whose size is proportional to $\epsilon$. Hence, we conclude $L^2$-global practical attractivity of the origin. Finally, using boundedness of $(V,\hat{\theta})$ and the structure of the feedback law $\kappa$, boundedness of the control input $u_1$ follows. Let us show that
\begin{align} \label{theta_bound}
\hat{\theta}(t) \leq \max\{\theta + C^2/\epsilon + \sigma + \Delta, \hat{\theta}(0) \} \quad \forall t \geq 0. 
\end{align}
To conclude that the inequality \eqref{theta_bound} is verified, we first suppose that 
$\hat{\theta}(0) \leq  \theta+C^2/\epsilon + \sigma$. 
As a result, either 
$$ \hat{\theta}(t) \leq  \theta+C^2/\epsilon + \sigma \quad \forall t \geq 0. $$
Otherwise, in view of \ref{item:RR1} in Algorithm \ref{adapt_0}, there exist $k^* \geq 1$ 
and $t \in [(k^*-1)\tau, k^* \tau)$
such that 
$$  \theta+C^2/\epsilon + \sigma < \hat{\theta}(t) \leq \theta+C^2/\epsilon + \sigma + \Delta. $$ 
Using \eqref{eqYoungafter}, we then conclude that
\begin{align*}
V(s) \leq V(t) \exp^{-\sigma (s-t)}+\frac{\epsilon}{\sigma} \quad  \forall s \in [t,k^* \tau].
\end{align*}
This implies that $\hat{\theta}$ is constant on $[t, k^*\tau]$. Since $\hat{\theta}$ is nondecreasing, it follows that  
\begin{align*}
V(s) \leq V(t) \exp^{-\sigma (s-t)}+\frac{\epsilon}{\sigma} \qquad  \forall s \geq t,
\end{align*}
which in turn
implies that 
$$\hat{\theta}(s) = \hat{\theta}(t) \leq \theta+C^2/\epsilon + \sigma + \Delta \qquad \forall s \geq t. $$
If 
 $\hat{\theta}(0) > \theta+C^2/\epsilon + \sigma, $ 
  we use the fact that $\hat{\theta}$ is nondecreasing to conclude that
\begin{align*}
V(s) \leq V(0) \exp^{-\sigma s}+\frac{\epsilon}{\sigma} \qquad  \forall s \geq 0,
\end{align*}
which in turn
implies that 
$ \hat{\theta}(s) 
= \hat{\theta}(0)$ for all $s \geq 0. $
To analyze the function $V$, we let $k \geq 0$ such that $\hat{\theta}$ is constant on $[k\tau, +\infty)$. As a result, for each $t \in [k\tau, +\infty)$, there exists $n \geq 0$ such that $t\in [(k+n)\tau,(k+1+n)\tau)$. 
Now, according to \ref{item:RR1}, we have
\begin{equation}
\label{final_V_o}
\begin{aligned}
V(t) & \leq V((k+n)\tau)\exp^{-\sigma (t-(k+n)\tau)} + \frac{\epsilon}{\sigma} 
\\ & \qquad \qquad \qquad  \forall t\in [(k+n)\tau, (k+1+n)\tau]. 
\end{aligned}
\end{equation}
Moreover, by continuity of $V$ and according to \ref{item:RR1}, we conclude that, for each $i\in \{1,...,n\}$, we have   
\begin{align*}
V((k+i)\tau)\leq V((k+i-1)\tau)\exp^{-\sigma \tau} + \frac{\epsilon}{\sigma}. 
\end{align*}
We show next that
\begin{align}
V((k+n)\tau) \leq V(k\tau)\exp^{-\sigma n \tau} + \frac{\epsilon}{\sigma(1-\exp^{-\sigma \tau})}. \label{final_V} 
\end{align}
Indeed, the latter inequality combined with \eqref{final_V_o} allows us to conclude $L^2$-global practical attractivity of the origin. To prove \eqref{final_V}, we note that 
\begin{align*}
V&((k+n)\tau)  \leq V((k+n-1)\tau)\exp^{-\sigma \tau} + \frac{\epsilon}{\sigma} \nonumber \\ &
\leq V((k+n-2)\tau)\exp^{-2\sigma \tau}  
+ \frac{\epsilon}{\sigma}\left(1+\exp^{-\sigma \tau}\right) 
\nonumber \\ & 
\leq V(k\tau)\exp^{-\sigma n \tau} + \frac{\epsilon}{\sigma}\left(\sum_{j=0}^{n}\exp^{-j\sigma \tau} \right).
\end{align*}
Finally, \eqref{final_V} follows using the fact that 
$$\sum_{j=0}^{n}\exp^{-j\sigma \tau}\leq \sum_{j=0}^{\infty}\exp^{-j\sigma \tau} 
\leq \frac{1}{1-\exp^{-\sigma \tau}}. $$
\end{proof}

\begin{remark}
    Under full sensing and when $f\equiv 0$, we can guarantee $L^2$ global exponential stability of the origin by following an approach analogous to the one  in this section. More specifically, since $f \equiv 0$, then $\dot{V}\leq (\theta-\hat{\theta})V$. The parameter $\hat{\theta}$ is then updated according to Algorithm \ref{adapt_0}, while setting $\epsilon := 0$ therein. The rest follows using the arguments in the proof of Theorem \ref{thm_p}.
\end{remark}

\section{Simulation Example} \label{numerical}

The system \eqref{KS} under \eqref{eqBKS} is simulated with the controller introduced in Section \ref{GUUB}. 
We set $\lambda(x) := 4\pi^{2}/0.25+50+2\sin(4x)^2$ for all $x \in (0,Y)$ and $\lambda(x) := \lambda (x-Y)$ for all $x\in (Y,1)$ and  $f(x,t) := 12\times 10^3\sin (2\times 10^4t)+\xi(x,t)$, where $\xi$ is a white noise of power $80$ dBW generated using the Matlab function 'wgn'. 
We set $I_{1} := [0,1)\cup [2,2.8)\cup [3.9,5)\cup [5.5,6.5)\cup [7,7.6) \times 10^{-3}$ and $I_{2} := [1,2)\cup [2.8,3.9)\cup [5,5.5)\cup [6.5,7)\cup [7.6,8) \times 10^{-3}$.  Finally, $(u_1,u_2)$ are given by \eqref{eqInputs}, with $\kappa$  defined in \eqref{eq:kappa} for $(\varepsilon, \delta) := (1,2)$, and $\hat{\theta}_1$ and $\hat{\theta}_2$ designed according to Algorithm \ref{adapt_algo_2} with $\hat{\theta}_1(0)=\hat{\theta}_2(0)=0$, $\Delta_1=\Delta_2 =0.01$, and $\sigma = 100$.

The system is discretized with the mesh-free collocation method in \cite{collocation}, based on radial basis functions (RBFs). We estimate the boundary values $w_{xxx}(0)$ and $w_{x}(0)$ using the Euler forward scheme and $v_{xxx}(1)$ and $v_{x}(1)$ using the Euler backward scheme. The Lyapunov function candidates $(V_{1},V_{2})$  are estimated using Riemannian sums.

We use multiquadric RBFs, which depend on a shape parameter $c :=0.4$.
The initial and final simulation times are selected as $t_{1} := 0$ and $t_{f} := 8 \times 10^{-3}$, respectively.  We select $10$ uniformly separated collocation points on the interval $[0,Y]$ (with $Y:=0.5$), ranging from $x_{o} := 0$ to $x_{9} := Y$, and select the same number of points on the interval $[Y,1]$. Hence, our spatial step is $\Delta x \approx 0.05$, which makes us choose the time step of $\Delta t := 10^{-7}$ to keep the ratio $\Delta t/\Delta x^4$  small. For comparison, the time step $\Delta t := 10^{-9}$ and the spatial step $\Delta x := 5\times 10^{-3}$ are used in \cite{kdv}, to simulate the third-order Korteweg-de Vries-Burgers' equation defined on $[0,1]$. Finally, the initial condition $u_o(x) := -A(\cos (4\pi x)-1)$ for all $x \in (0,1)$ is tested for different values of $A>0$. 

Figure \ref{fig2} shows the closed-loop solutions corresponding to the initial condition $u_o$ under $A = 3$. The corresponding control inputs are shown in Figure \ref{control_fig}. 
As expected, the inputs are bounded, at the price of some chattering 
due to the discontinuous nature of our controller. 
Furthermore, the corresponding plots of $t \mapsto (V_1(t),V_2(t))$ and $t \mapsto V_1(t)+V_2(t)$ are depicted in Figure \ref{fig3}. On every interval in $I_2$, where $u_1 = 0$, the final value of $V_1$ is greater than its initial value at the beginning of that time interval. This increase is compensated over the next intervals in $I_1$, where $u_1 = \kappa (V_1,w_{xxx}(0),\hat{\theta}_1)$.  
The same behavior is observed for $V_2$, which increases over intervals in $I_1$ and decreases over intervals in $I_2$. 
In Figure \ref{fig_adapt1}, we plot the evolution of the adaptation parameters $t \mapsto (\hat{\theta}_{1}(t),\hat{\theta}_{2}(t))$, which shows the existence of an adaptation phase, during which $\hat{\theta}_1$ and $\hat{\theta}_2$ increase before they become, as expected, constant. Finally, Figure \ref{last_fig} illustrate  the evolution of $t \mapsto V_{1}(t) + V_{2}(t)$ for the initial condition $u_o$ with $A \in 
\{2,3,4,5,6,7\}$. The conclusions of Theorem \ref{theorem2} are in agreement with Figure \ref{last_fig}, which shows finite-time convergence of $V_{1}+V_{2}$ to a bound that is independent of the various choices of $A$.

\begin{figure}
\centering
\includegraphics[scale=0.19]{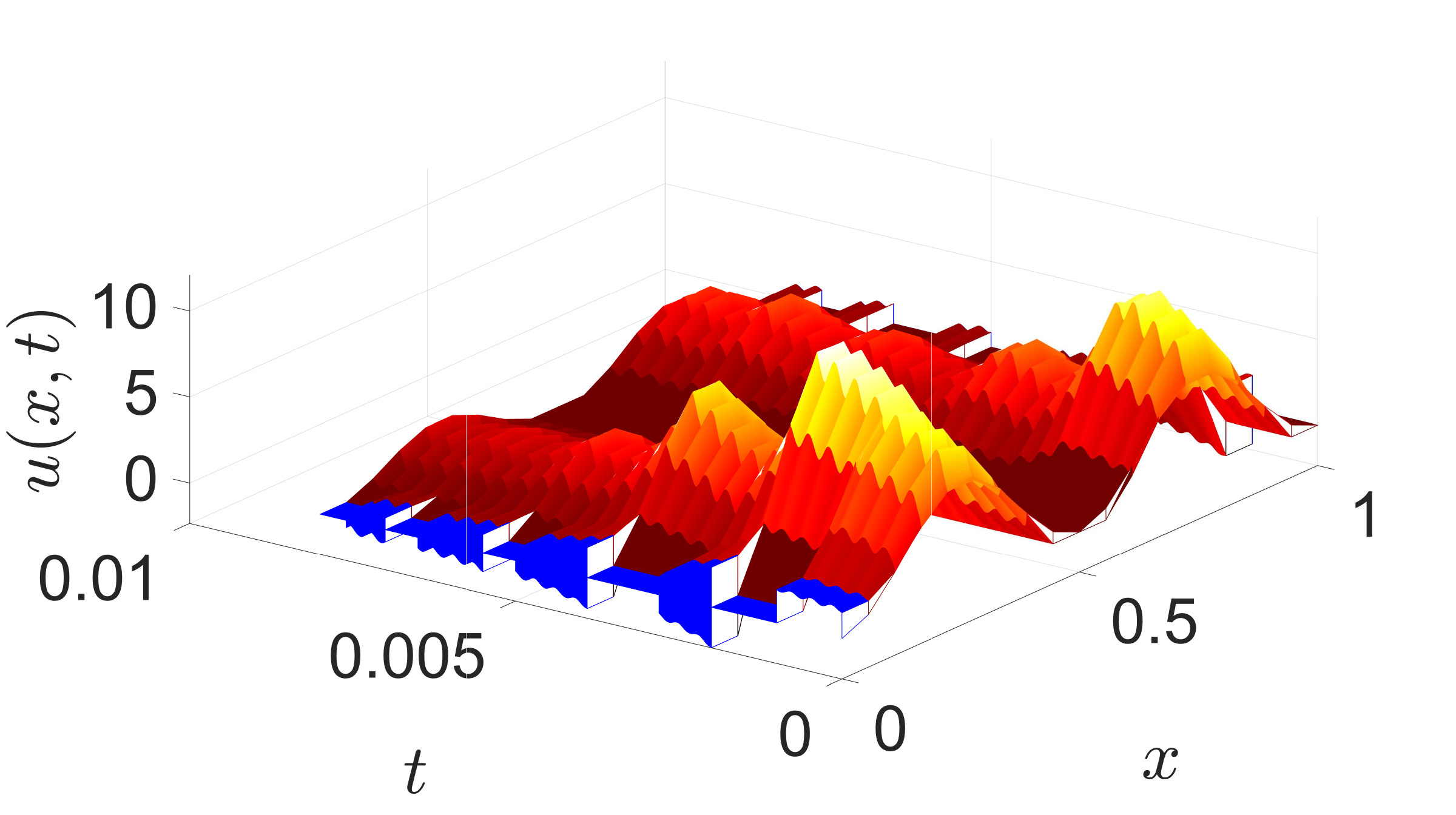}
\caption{The KS response to \eqref{eqBKS}, \eqref{eqInputs}, and Algorithm \ref{adapt_algo_2}.}
\label{fig2}
\end{figure}

\begin{figure}
    \centering
    \includegraphics[scale=0.19]{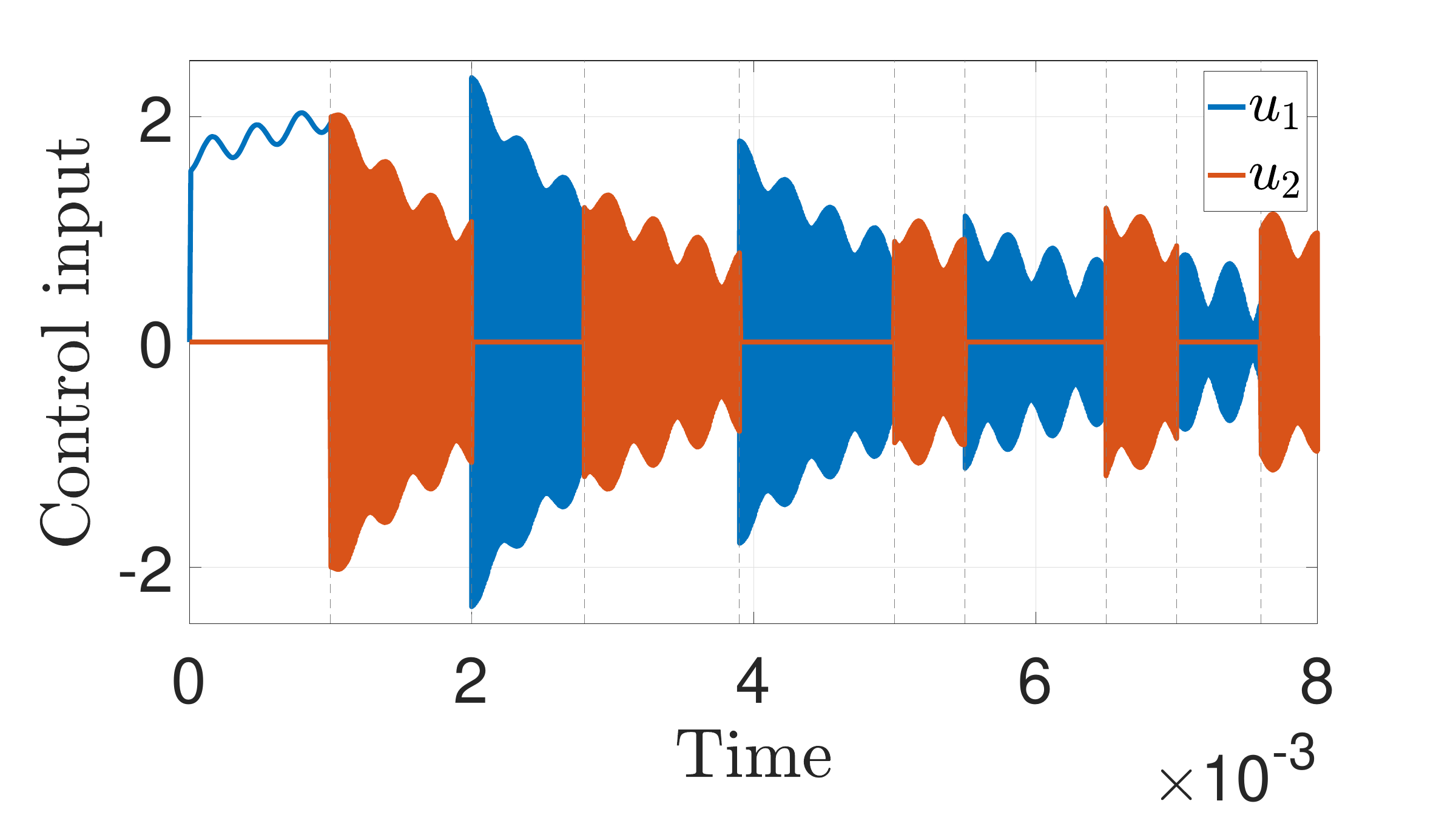}
    \caption{The  inputs $(u_1,u_2)$ in  \eqref{eqInputs} and Algorithm \ref{adapt_algo_2}.}
    \label{control_fig}
\end{figure}
\begin{figure}
    \centering
    \includegraphics[scale=0.19]{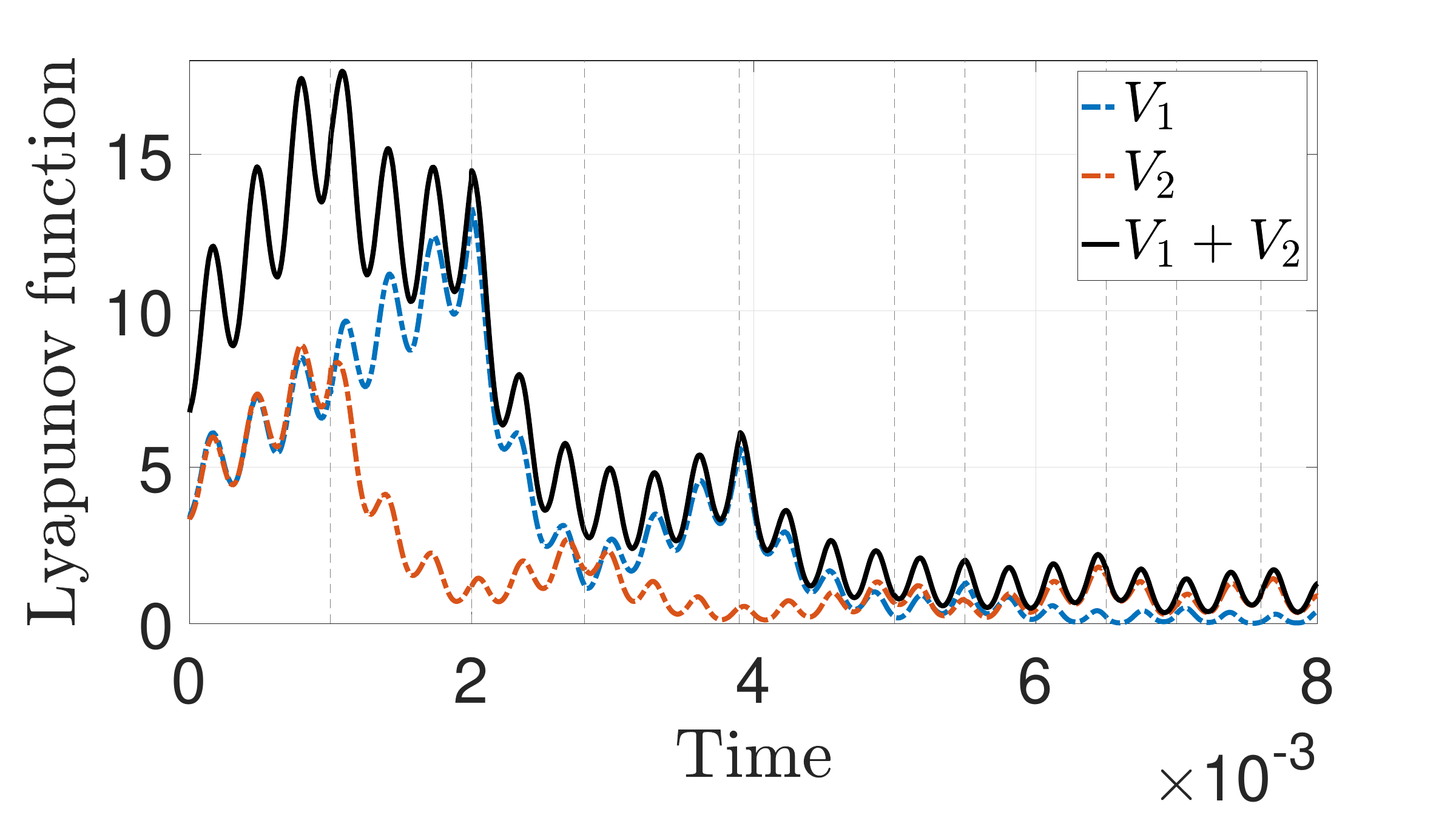}
    \caption{Lyapunov functions $(V_1,V_2)$ and $V_1+V_2$ along the KS response under \eqref{eqBKS}, \eqref{eqInputs}, and Algorithm \ref{adapt_algo_2}.}
    \label{fig3}
\end{figure}
 \begin{figure}
     \centering
     \includegraphics[scale=0.19]{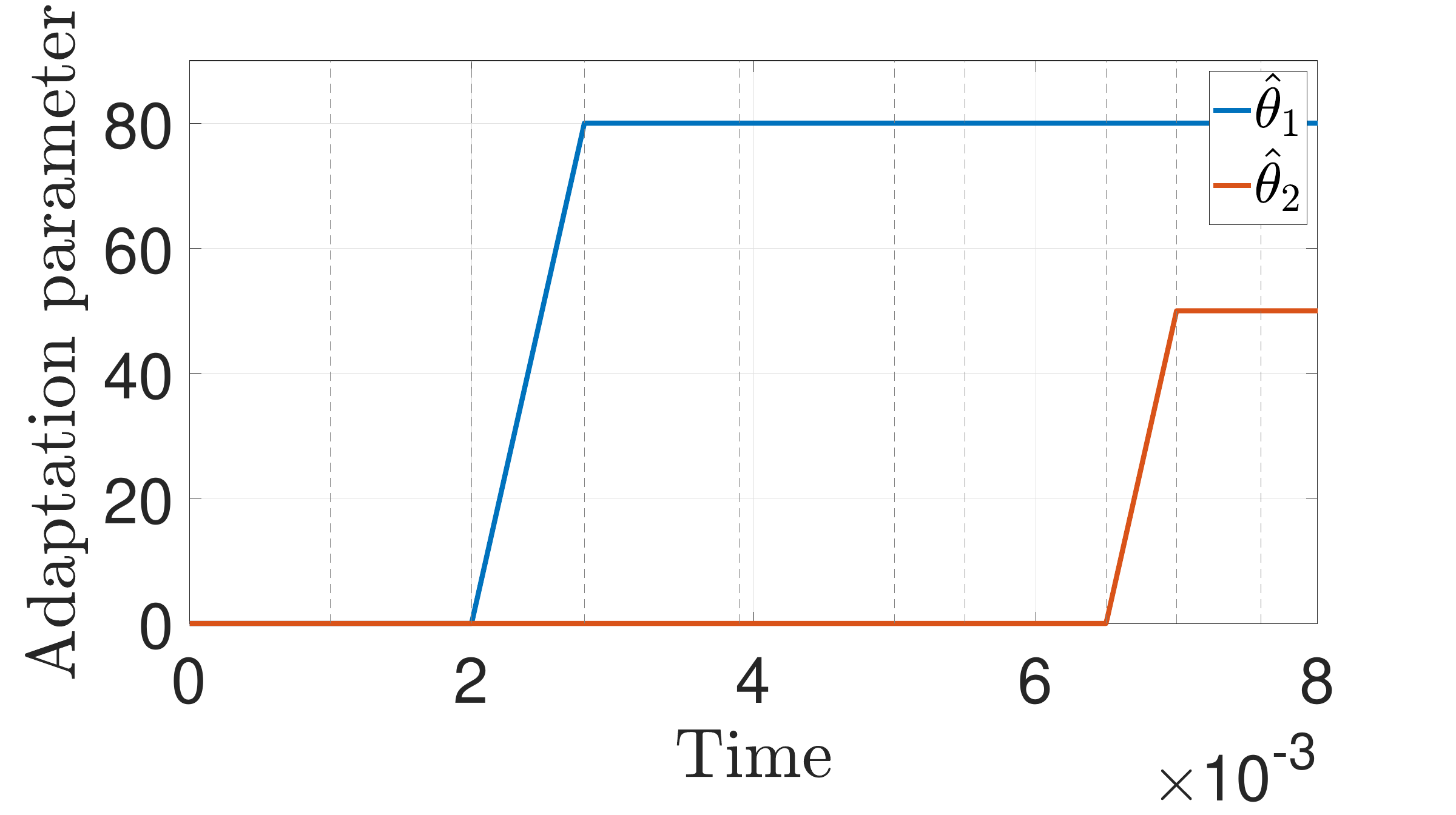}
     \caption{The adaptation parameters $(\hat{\theta}_{1},\hat{\theta}_{2})$ constructed according to Algorithm \ref{adapt_algo_2}.}
     \label{fig_adapt1}
 \end{figure}

\begin{figure}
    \centering
    \includegraphics[scale=0.19]{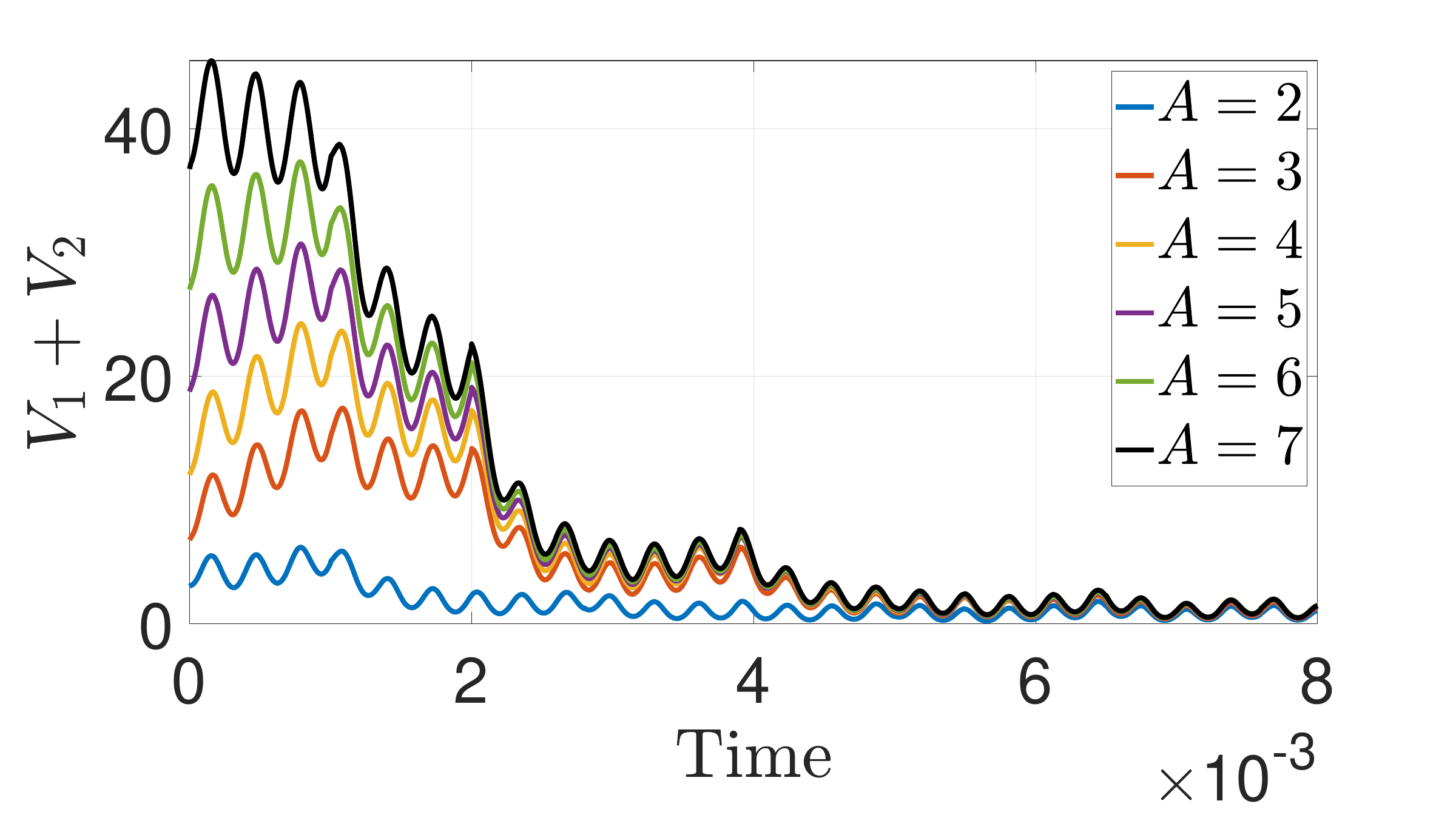}
    \caption{$V_1 + V_2$ along the KS response to \eqref{eqBKS}, \eqref{eqInputs}, and Algorithm \ref{adapt_algo_2}, for different sizes of initial condition.}
    \label{last_fig}
\end{figure}

\section{Conclusion}
In this paper, we studied the boundary stabilization problem for the perturbed KS equation under an intermittent sensing scenario, and compared the obtained results with those we obtained with full sensing. The intermittent sensing scenario forces us to consider in-domain conditions, that we do not necessarily require under full sensing. Our design is Lyapunov-based, tailored to some adaptive design to handle unknown parameters and perturbations. The obtained results are summarized in the following table.  
\begin{table}[h!]
    \centering
    \begin{tabular}{@{}llll@{}}
        \toprule
        \textbf{Sensing Scenario}  & \textbf{$f=0$} & \textbf{$\bar{f}$ known} & \textbf{$\bar{f}$ unknown} \\ 
        \midrule
        Full Sensing               & $L^2$-GES  & $L^2$-GpA & $L^2$-GpA         \\
        Intermittent Sensing       & $L^2$-GES  & $L^2$-ISS  & $L^2$-GUUB        \\ 
        \bottomrule
    \end{tabular}
\end{table}

Several challenging research directions emerge from our work. A primary extension would involve considering measurements taken intermittently over an arbitrary number of subregions, rather than just two. Another compelling direction concerns uncertainty in the location $Y$, which becomes particularly relevant when the Lyapunov functions $(V_1,V_2)$ must be approximated numerically using sensor measurements. A related scenario arises when enforcing the in-domain condition at $x=Y$ while gathering measurements over intervals $[0,Y_o]$ and $[Y_o,1]$ with $Y_o<Y$. The consideration of packet losses under an average dwell-time condition is also an interesting research perspective. Finally, extending our approach to accommodate measurements from scanning and pointwise sensors \cite{scanning,intermittent_app7} represents a natural progression of our work.

\par\noindent 
\parbox[t]{\linewidth}{
\noindent\parpic{\includegraphics[height=1.5in,width=1in,clip,keepaspectratio]{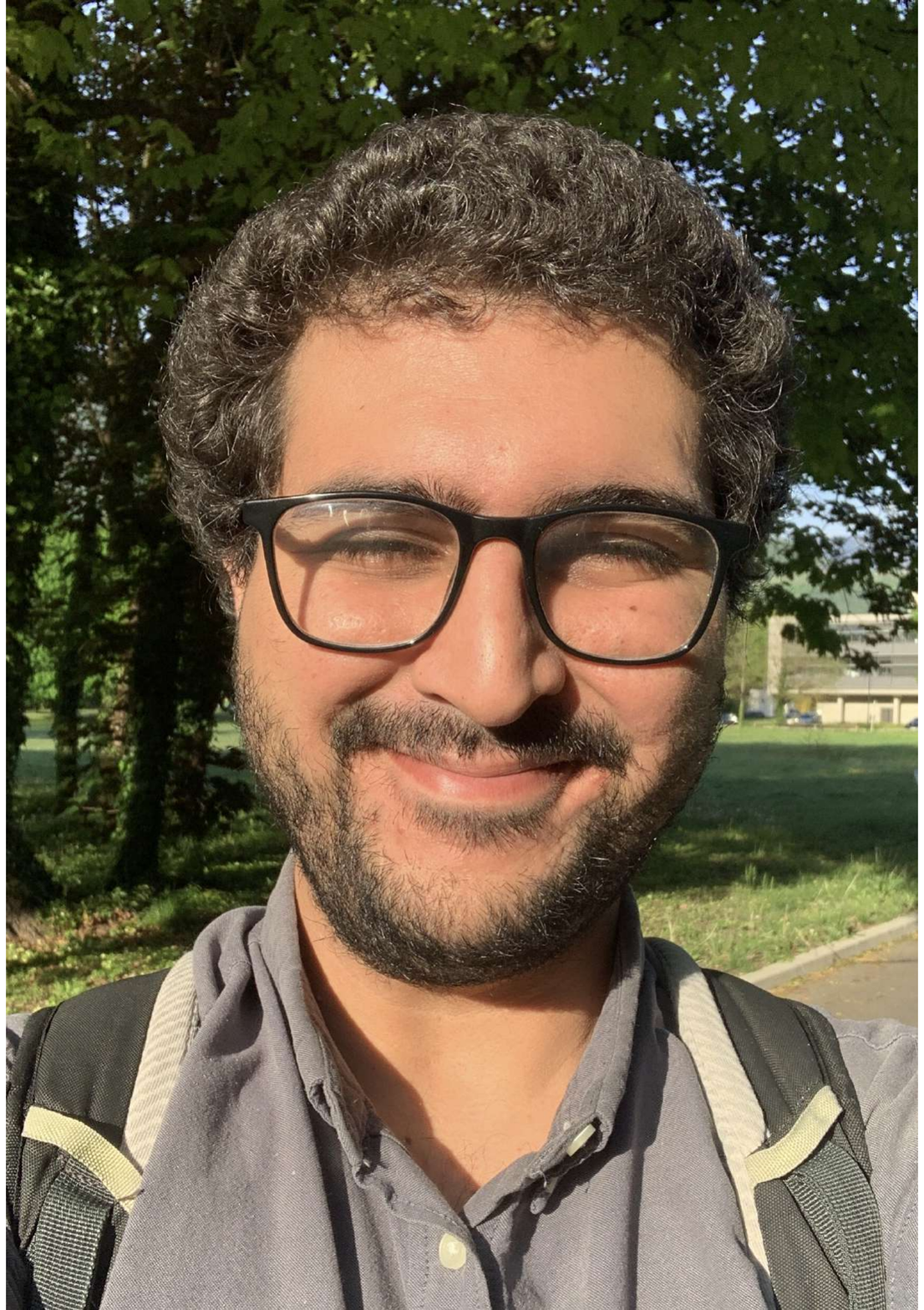}}
\noindent {\bf Mohamed Camil Belhadjoudja}\
received a Control-Engineer and MS degrees from the National Polytechnic School, Algiers, Algeria, in 2021, and a MS degree in Control Theory from the Université Grenoble Alpes, Grenoble, France, in 2022. He is currently a Ph.D. student at the Grenoble Images Parole Signal Automatique Laboratory (GIPSA-lab), Grenoble, France. His research is focused on control and state estimation for nonlinear partial differential equations.}


\par\noindent 
\parbox[t]{\linewidth}{
\noindent\parpic{\includegraphics[height=1.5in,width=1in,clip,keepaspectratio]{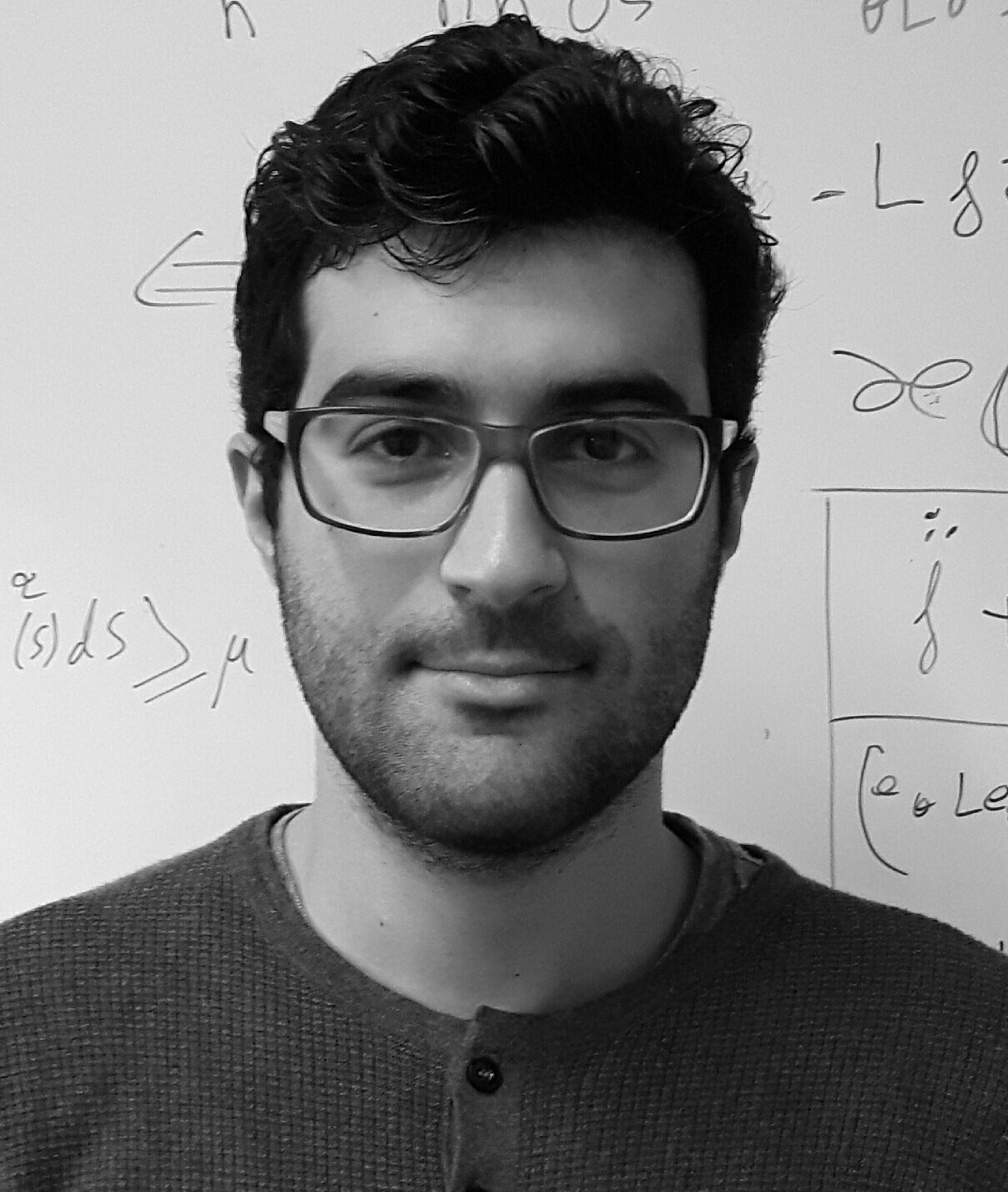}}
\noindent {\bf Mohamed Maghenem}\
received his Control-Engineer degree from the Polytechnical School of Algiers, Algeria, in 2013, his  PhD degree on Automatic Control from the University of Paris-Saclay, France, in 2017.  He was a Postdoctoral Fellow at the Electrical and Computer Engineering Department at the University of California at Santa Cruz from 2018 through 2021. M. Maghenem holds a research position at the French National Centre of Scientific Research (CNRS) since January 2021.  His research interests include control systems theory (linear,  non-linear,  and hybrid) to ensure (stability, safety, reachability, synchronisation, and robustness); with applications to mechanical systems,  power systems,  cyber-physical systems,  and some partial differential equations.}
 

\par\noindent 
\parbox[t]{\linewidth}{
\noindent\parpic{\includegraphics[height=1.5in,width=1in,clip,keepaspectratio]{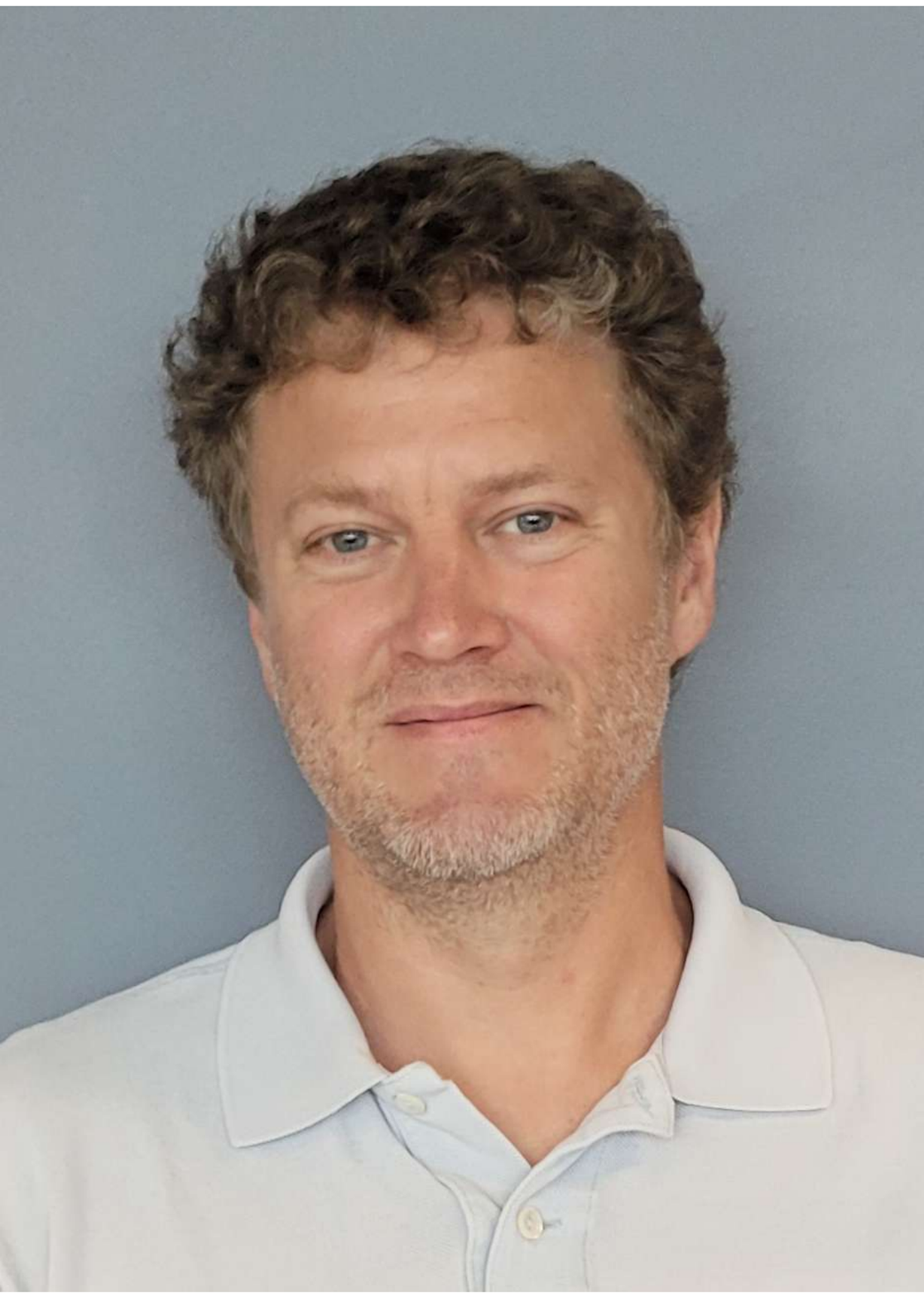}}
\noindent {\bf Emmanuel Witrant}\
obtained a B.Sc. in Aerospace Eng. from Georgia Tech in 2001 and a Ph.D. in Automatic Control from Grenoble University in 2005. He joined the Physics department of Univ. Grenoble Alpes and GIPSA-lab as an Associate Professor in 2007 and became Professor in 2020. His research interest is focused on finding new methods for modeling and control of inhomogeneous transport phenomena (of information, energy, gases...), with real-time and/or optimization constraints. Such methods provide new results for automatic control (time-delay and distributed systems), controlled thermonuclear fusion (temperature/magnetic flux profiles in tokamak plasma), environmental sciences (atmospheric history of trace gas from ice cores measurements) and Poiseuille’s flows (ventilation control in mines, car engines and intelligent buildings).}

\par\noindent 
\parbox[t]{\linewidth}{
\noindent\parpic{\includegraphics[height=1.5in,width=1in,clip,keepaspectratio]{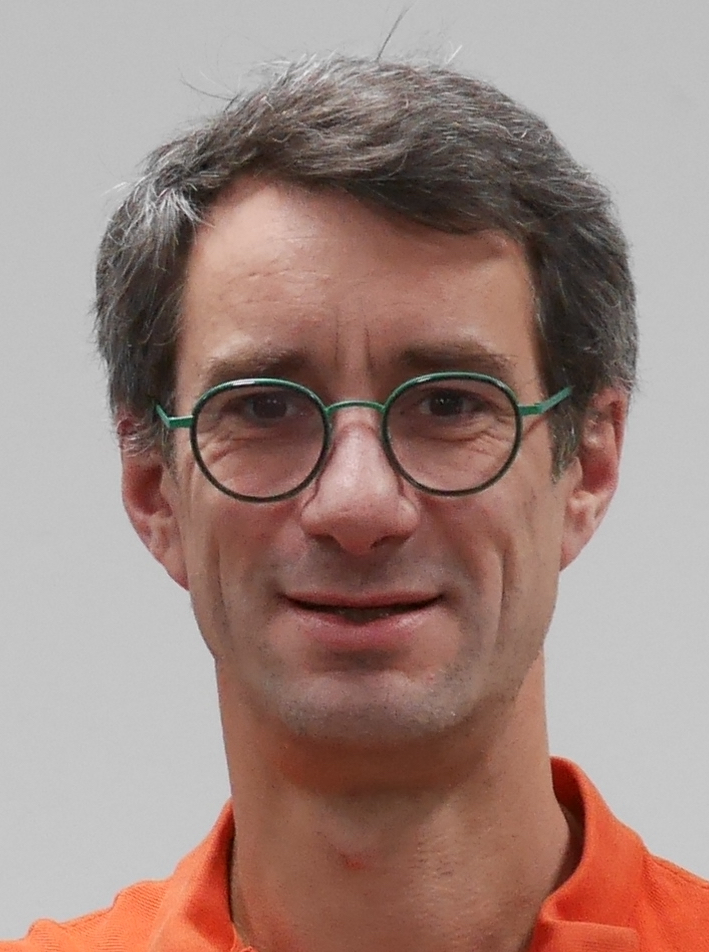}}
\noindent {\bf Christophe Prieur}\
graduated in Mathematics from the Ecole Normale Supérieure de Cachan, France in 2000. He received the Ph.D. degree in 2001 in Applied Mathematics from Université Paris-Sud, France. From 2002 he was a researcher at CNRS at the laboratory SATIE, Cachan, France, and at the LAAS, Toulouse, France (2004–2010). In 2010 he joined the Gipsa-lab, Grenoble, France where he is currently a senior researcher at CNRS (since 2011). 
His current research interests include nonlinear control theory, hybrid systems, and control of partial differential equations, with applications including navigation and object tracking, fluid dynamics, and fusion control. He is an IMA Fellow, and an IEEE Fellow.}
\vspace{4\baselineskip}

\newpage 

\ifitsdraft 
\section*{Appendix}

\begin{lemma} [\cite{PITT}, page 84] \label{halperin-pitt}
Given $u:(a,b)\rightarrow \mathbb{R}$ be twice continuously differentiable and let $\epsilon > 0$. Then,  
\begin{equation} \label{halperin}
\begin{aligned}
\int_{a}^{b}u_x^2dx \leq&~ \left[ \frac{1}{\epsilon}  +\frac{12}{(b-a)^2} \right] \int_{a}^{b} u^2dx + \epsilon \int_{a}^{b} u_{xx}^2dx.
\end{aligned}
\end{equation}
\end{lemma}

\begin{lemma}\label{to_use_proofs1}
Let $\{t_{i}\}_{i=1}^{\infty} \subset \mathbb{R}$, with $t_1=0$ and $t_{i+1}>t_i$, $I_1 := \bigcup_{k=1}^{\infty}[t_{2k-1},t_{2k})$, and $I_2:= \bigcup_{k=1}^{\infty}[t_{2k},t_{2k+1})$. Moreover, let $\underline{T}_1\leq t_{2k}-t_{2k-1}\leq \bar{T}_1$ and $\underline{T}_2\leq t_{2k+1}-t_{2k}\leq \bar{T}_2$ for all $k\geq 1$, and for some constants $\underline{T}_1, \underline{T}_2,\bar{T}_1,\bar{T}_2>0$. Furthermore, let $\theta, C\in \mathbb{R}_{\geq 0}$, and $t\mapsto \hat{\theta}(t)\geq 0$ be a non-decreasing function such that there exists $T\geq 0$, such that, for all $t\geq T$, we have 
\begin{align}
\hat{\theta}(t) \geq \theta + C + \frac{(\theta+1)\bar{T}_2+\sigma (\bar{T}_1+\bar{T}_2)}{\underline{T}_1}+1, \label{15bis}
\end{align}
where $\sigma >0$. We let $V:\mathbb{R}_{\geq 0}\to \mathbb{R}_{\geq 0}$ be a locally absolutely continuous solution to the switched system of differential inequalities 
\begin{equation} \label{lem4_switch}
\begin{aligned}
\left\{
\begin{matrix}
\begin{matrix}
\dot{V} \leq & (\theta-\hat{\theta})V + C\sqrt{V}
\end{matrix}
\qquad \text{a.e. on}~ I_{1},
\\ 
\begin{matrix}
\dot{V} \leq & \theta V+C\sqrt{V}
\end{matrix}
\qquad \qquad \ \ \text{a.e. on}~ I_{2}.
\end{matrix}
\right.
\end{aligned}
\end{equation}
Then, for all $k>1$ such that $t_{2k-3}\geq T$, we have 
\begin{align}
&V(t_{2k-2})\leq V(t_{2k-3})\exp^{-\sigma (t_{2k-2}-t_{2k-3})} + C, \label{pre_17} \\
&V(t_{2k-2})\leq V(t_{2k-3})\exp^{-\sigma (t_{2k-2}-t_{2k-3})}+\hat{\theta}(t_{2k-3}).\label{17}
\end{align}
\end{lemma}
\begin{proof}
To prove \eqref{17}, first note that, a.e. on $[t_{2k-3},t_{2k-2})\subset I_1$, we have, according to \eqref{lem4_switch},
\begin{align}
\dot{V} &\leq (\theta - \hat{\theta})V + C \sqrt{V} \nonumber \\
& \leq (\theta+C-\hat{\theta})V + C, \label{18}
\end{align}
where we have used the fact that $\sqrt{V}\leq V+1$. Integrating \eqref{18}, we obtain 
\begin{align*}
V(t_{2k-2}) & \leq V(t_{2k-3})\exp^{\int_{t_{2k-3}}^{t_{2k-2}}(\theta+C-\hat{\theta}(t))dt} \\
&+ C\int_{t_{2k-3}}^{t_{2k-2}}\exp^{\int_{\tau}^{t_{2k-2}}(\theta+C-\hat{\theta}(s))ds} d\tau. 
\end{align*}
Furthermore, since $t \mapsto \hat{\theta}(t)$ is non-decreasing, we conclude that $t \mapsto \theta + C -\hat{\theta}(t)$ is non-increasing. Therefore, 
$$ \theta + C -\hat{\theta}(t) \leq  \theta + C -\hat{\theta}(t_{2k-3}),\qquad \forall t \geq t_{2k-3}.  $$
Using the latter inequality, we can write  
\begin{align*}
V(t_{2k-2}) & \leq V(t_{2k-3})\exp^{(\theta+C-\hat{\theta}(t_{2k-3}))(t_{2k-2}-t_{2k-3})} \\
&+ C\int_{t_{2k-3}}^{t_{2k-2}}\exp^{(\theta+C-\hat{\theta}(t_{2k-3}))(t_{2k-2}-\tau)}d\tau \\ &
\leq V(t_{2k-3})\exp^{(\theta+C-\hat{\theta}(t_{2k-3}))(t_{2k-2}-t_{2k-3})} \\
&+ C\bigg[\frac{\exp^{(\theta+C-\hat{\theta}(t_{2k-3}))(t_{2k-2}-\tau)}}{-(\theta+C-\hat{\theta}(t_{2k-3}))}\bigg]_{\tau = t_{2k-3}}^{\tau = t_{2k-2}}\\ &
\leq V(t_{2k-3})\exp^{(\theta+C-\hat{\theta}(t_{2k-3}))(t_{2k-2}-t_{2k-3})} 
\\ &
- \bigg[\frac{C}{(\theta+C-\hat{\theta}(t_{2k-3}))}\bigg]
 \\
 &+
C\bigg[\frac{\exp^{(\theta+C-\hat{\theta}(t_{2k-3}))(t_{2k-2}-t_{2k-3})}}{(\theta+C-\hat{\theta}(t_{2k-3}))}\bigg].
\end{align*}
Next, in view of inequality \eqref{15bis}, we conclude that 
\begin{align}
\left[\theta + C -\hat{\theta}(t_{2k-3}) \right] \underline{T}_1 & \leq - \underline{T}_1 - (\theta+1)\bar{T}_2-\sigma (\bar{T}_1+\bar{T}_2)\nonumber \\
&< -\sigma (\bar{T}_1+\bar{T}_2).\label{27}
\end{align}
and 
$$\theta+C-\hat{\theta}(t_{2k-3}) \leq -1.$$
Now, using the fact that 
$$ \underline{T}_1\leq t_{2k-2}-t_{2k-3} \quad  \text{and} \quad \theta+C-\hat{\theta}(t_{2k-3})  < 0,  $$ 
we obtain
\begin{align*}
&\exp^{(\theta+C-\hat{\theta}(t_{2k-3}))(t_{2k-2}-t_{2k-3})}  \leq \exp^{(\theta+C-\hat{\theta}(t_{2k-3}))\underline{T}_1},
\\
&C\bigg[\frac{\exp^{(\theta+C-\hat{\theta}(t_{2k-3}))(t_{2k-2}-t_{2k-3})}}{(\theta+C-\hat{\theta}(t_{2k-3}))}\bigg] < 0,
\end{align*}
and thus 
\begin{align*}
V(t_{2k-2}) & 
\leq V(t_{2k-3})\exp^{(\theta+C-\hat{\theta}(t_{2k-3}))\underline{T}_1} 
\\ &
- \bigg[\frac{C}{(\theta+C-\hat{\theta}(t_{2k-3}))}\bigg]
 \\
 &+
C\bigg[\frac{\exp^{(\theta+C-\hat{\theta}(t_{2k-3}))(t_{2k-2}-t_{2k-3})}}{(\theta+C-\hat{\theta}(t_{2k-3}))}\bigg]
\\ & 
\leq V(t_{2k-3})\exp^{(\theta+C-\hat{\theta}(t_{2k-3}))\underline{T}_1} \\
&- \bigg[\frac{C}{(\theta+C-\hat{\theta}(t_{2k-3}))}\bigg].
\end{align*}
Furthermore, in view of \eqref{27} and, since $t_{2k-2}-t_{2k-3} \leq \bar{T}_1+\bar{T}_2$, we conclude that 
\begin{align*}
\exp^{(\theta+C-\hat{\theta}(t_{2k-3}))\underline{T}_1} &\leq \exp^{-\sigma (\bar{T}_1+\bar{T}_2)} \\&\leq \exp^{-\sigma (t_{2k-2}-t_{2k-3})}.
\end{align*}
Combining the aforementioned inequalities,  and the fact that $-\frac{C}{\theta+C-\hat{\theta}(t_{2k-3})} \leq C$, we obtain
\begin{align*}
V(t_{2k-2}) & \leq V(t_{2k-3})\exp^{(\theta+C-\hat{\theta}(t_{2k-3}))\underline{T}_1} + C \nonumber\\
& \leq V(t_{2k-3})\exp^{-\sigma (t_{2k-2}-t_{2k-3})} + C. 
\end{align*}
Finally, \eqref{17} is obtained by combining \eqref{pre_17} with the fact that, from \eqref{15bis}, we have $\hat{\theta}(t_{2k-3})\geq C$.
\end{proof}
\begin{lemma}\label{ineqsquare}
Let $V : \mathbb{R}_{\geq 0} \rightarrow \mathbb{R}_{\geq 0}$ be a locally absolutely continuous solution to the differential inequality 
    \begin{equation}
        \dot{V}\leq \theta V+C\sqrt{V} \quad \text{a.e. on} \ [0,T] \subset \mathbb{R}_{\geq 0}, \label{lem31}  
    \end{equation}
    where $\theta , C\geq 0$ are constants. Then, for any constant $\delta >0$ we have, for all $t\in [0,T]$,
    \begin{equation}
    \hspace{-0.2cm} V(t)\leq V(0)\exp^{(\theta+\delta)t} + \frac{C^{2}/4\delta}{ (\theta + \delta)}\left(\exp^{(\theta + \delta)t}-1\right).\label{lem32}
    \end{equation}
\end{lemma}
\begin{proof}
Let $\delta >0$ and consider the function $f(V) := \sqrt{V}-\frac{\delta V}{C}-\frac{C}{4\delta }$. By differentiating $f$, we find for all $V>0$, $ f'(V) = \frac{1}{2\sqrt{V}}-\frac{\delta }{C}$. As a result, the function $f$ is strictly increasing on $[0,C^{2}/(4\delta^{2})]$ and strictly decreasing on $[C^{2}/(4\delta^{2}),\infty)$. Moreover, $f(0)=-C/4\delta$, and $f(C^{2}/(4\delta^{2}))=0$. Therefore, for all $V\geq 0$, we have $f(V)\leq 0$. As a consequence, we can rewrite \eqref{lem31} as $\dot{V} \leq \theta V+C\sqrt{V}
\leq (\theta+\delta)V+\frac{C^{2}}{4\delta}$. By integrating this inequality from $0$ to $t$, \eqref{lem32} follows. 
\end{proof}

\begin{lemma}\label{last_lem}
Let the function $V:\mathbb{R}_{\geq 0} \to \mathbb{R}_{\geq 0}$ be locally absolutely continuous, 
let a sequence 
$\{T_i\}_{i=0}^{\infty}$ and $\underline{T}, \overline{T} > 0$ 
such that 
$ T_{0}=0 ~~ \text{and} ~~ \overline{T} \geq  T_{i+1} - T_i \geq \underline{T}  \quad \forall i \in \mathbb{N}$. Let $\{i_1, i_2, ..., i_{N^*}\} \subset \mathbb{N}$, with $N^* \in \mathbb{N}$, and let $(M,\psi,\sigma)$ be nonnegative constants. Assume that
\begin{itemize}
    \item For each $i \in \{i_1, i_2, ..., i_{N^*}\}$,
    \begin{align} \label{eq1lemma17}
\hspace{-1cm} V(T_{i+1}) \leq \left(V(T_{i})+\frac{M^2}{4}\right)\exp^{\psi(T_{i+1} - T_{i})}.
    \end{align}
    \item For each $i \in \mathbb{N}/
    \{i_1, i_2, ..., i_{N^*}\}$,
    \begin{equation}
\label{eq2lemma17}
    \begin{aligned}
   \hspace{-0.4cm} V(T_{i+1}) & \leq V(T_{i})\exp^{-\sigma (T_{i+1}-T_{i})}  
    \\ \hspace{-0.4cm} &
    +\left(M+\frac{M^2}{4}\right) \exp^{\psi(T_{i+1}-T_{i})}. 
    \end{aligned}
    \end{equation}
\end{itemize}
 Then, we have
\begin{align*}
&V(T_{i}) \leq \exp^{(\sigma+\psi) N^*\overline{T}}V(0)\exp^{-\sigma T_i} \\
&~+ \left(\frac{\exp^{(\sigma+\psi) N^*\overline{T}}}{1-\exp^{-\sigma \underline{T}}} \right)\frac{4M+M^2}{4} \exp^{\psi \overline{T}}, \quad \forall i \in \mathbb{N}. 
\end{align*}
\end{lemma}

\begin{proof}
To prove the Lemma, it is enough to show that
\begin{equation}
\label{to_prove_final}
\begin{aligned}
V(T_i) & \leq \exp^{(\sigma+\psi) N^*\overline{T}}
\left(
V(0) \exp^{-\sigma T_i}  
+ \eta(i) \right) ~~ \forall i \in \mathbb{N}, 
\end{aligned}
\end{equation}
where, for each $i \in \mathbb{N}$,  
\begin{align*}
\eta(i) & := 
\left(\sum_{k=0}^{i}\exp^{-k\sigma \underline{T}} \right)\frac{4 M + M^2}{4} \exp^{\psi \overline{T}} \nonumber \\ &
\leq \left(\frac{1}{1-\exp^{-\sigma \underline{T}}} \right)\frac{4M+M^2}{4} \exp^{\psi \overline{T}}.
\end{align*}
To prove \eqref{to_prove_final}, it is sufficient to show that
\begin{equation}
\label{new_ni}
\begin{aligned}
 V(T_i) & \leq \exp^{(\sigma+\psi) N(i) \overline{T}}
\left( V(0)\exp^{-\sigma T_i} + \eta(i) \right) ~~~ \forall i \in \mathbb{N},
\end{aligned}
\end{equation}
where $ N(i) := \text{card} \{ [T_j,T_{j+1}] : j+1 \leq i,~j \in \{i_1,i_2,...,i_{N^*}\} \}$ 
is the number of time intervals $[T_j,T_{j+1}]$, $j \in \{i_1,i_2,...,i_{N^*}\}$, prior to $T_i$, which satisfies $N(i) \leq N^*$  for all $i \in \mathbb{N}$. To show \eqref{new_ni}, we proceed by recurrence. Indeed, for $i=0$, the inequality in \eqref{new_ni} is trivially satisfied. Suppose now that the inequality in \eqref{new_ni} is verified for $i \in \mathbb{N}$ 
and let us show that it is also verified for $i+1$.

Note that either $N(i+1) = N(i)$ or $N(i+1) = N(i)+1$. If $N(i+1)=N(i)$ then, using \eqref{eq2lemma17}, we obtain 
\begin{align*}
 V(T_{i+1}) & \leq V(T_{i}) \exp^{-\sigma (T_{i+1}-T_{i})} 
 \\ & + \left(M+\frac{M^2}{4}\right)\exp^{\psi(T_{i+1}-T_i)} \nonumber \\
 &\leq \bigg(\exp^{(\sigma + \psi)N(i)\bar{T}}(V(0)\exp^{-\sigma T_i} \\
 &+\eta (i))\bigg)\exp^{-\sigma (T_{i+1}-T_i)}
 \nonumber \\
 &~ ~ + \bigg(M+\frac{M^2}{4}\bigg)\exp^{\psi (T_{i+1}-T_i)} \nonumber \\
 &\leq \exp^{(\sigma + \psi)N(i)\bar{T}} V(0) \exp^{-\sigma (T_{i+1}-T_i+T_i)} \nonumber \\
 &~+ \eta (i) \exp^{(\sigma + \psi)N(i)\bar{T}}\exp^{-\sigma (T_{i+1}-T_i)} \nonumber \\
&~+ \bigg(M+\frac{M^2}{4}\bigg)\exp^{\psi (T_{i+1}-T_i)}\nonumber \\
& \leq \exp^{(\sigma +\psi) N(i) \overline{T}} V(0) \exp^{-\sigma T_{i+1}} \nonumber \\
& + \left( \eta(i) \exp^{-\sigma (T_{i+1}-T_i)} \right) \exp^{(\sigma+\psi) N(i) \overline{T}}  
\\
&
+\left(M+\frac{M^2}{4}\right)  \exp^{\psi (T_{i+1}-T_i)}.
\end{align*}
Using the fact that 
\begin{align*} 
 \eta(i) & \exp^{-\sigma (T_{i+1}-T_i)} \leq \bigg[\sum_{k=1}^{i+1}\exp^{-k\sigma \underline{T}}\bigg] \left(M+\frac{M^2}{4}\right) \exp^{\psi \overline{T}}, 
\end{align*}
we obtain
\begin{align*}
V & (T_{i+1}) \leq \exp^{(\sigma + \psi ) N(i) \overline{T}} V(0) \exp^{-\sigma T_{i+1}} \nonumber \\ & + \bigg[\sum_{k=1}^{i+1}\exp^{-k\sigma \underline{T}}\bigg]
\left(M+\frac{M^2}{4}\right) \exp^{\psi  \overline{T}} 
\exp^{(\sigma+\psi ) N(i) \overline{T}} 
\nonumber \\
&  + \left(M + \frac{M^2}{4} \right) \exp^{\psi (T_{i+1}-T_i)}.
\end{align*}
It implies that
\begin{align*}
V(T_{i+1}) & \leq 
\exp^{(\sigma + \psi ) N(i) \overline{T}} V(0) \exp^{-\sigma T_{i+1}} 
\nonumber \\ & 
+ \bigg[\sum_{k=1}^{i+1}\exp^{-k\sigma \underline{T}}\bigg]
\left(M+\frac{M^2}{4}\right) \exp^{\psi  \overline{T}} 
\exp^{(\sigma+\psi ) N(i) \overline{T}} 
\nonumber \\ &
+ \left(M + \frac{M^2}{4} \right) \exp^{\psi \bar{T}} \exp^{(\sigma+\psi) N(i) \overline{T}}. 
\end{align*}
Combining the latter two terms, we obtain
\begin{align*}
& V(T_{i+1}) \leq 
\exp^{(\sigma + \psi ) N(i) \overline{T}} V(0) \exp^{-\sigma T_{i+1}} \nonumber \\ & + \bigg[\sum_{k=0}^{i+1}\exp^{-k\sigma \underline{T}}\bigg]
\left(M+\frac{M^2}{4}\right) \exp^{\psi  \overline{T}} 
\exp^{(\sigma+\psi ) N(i) \overline{T}}. 
\end{align*}

Finally, since $N(i) = N(i+1) \leq N^*$
and  
\begin{align*}
& \eta(i+1) = \bigg[\sum_{k=0}^{i+1}\exp^{-k\sigma \underline{T}}\bigg]
\left(M+\frac{M^2}{4}\right) \exp^{\psi  \overline{T}},
\end{align*}
we obtain 
\begin{align*}
V (T_{i+1}) &  \leq\exp^{(\sigma + \psi )N(i+1) \overline{T}}V(0)\exp^{-\sigma T_{i+1}} \nonumber \\
& + \eta(i+1) \exp^{(\sigma + \psi ) N(i) \overline{T}}.
\end{align*}

If $N(i+1) = N(i)+1$, we use \eqref{eq1lemma17} to conclude that 
\begin{align*}
& V(T_{i+1})  \leq \left(V(T_i)+\frac{M^2}{4}\right)\exp^{\psi \overline{T}}  
\\ &
\leq \exp^{(\sigma + \psi )N(i) \overline{T}} \exp^{\psi \overline{T}} 
\\ &
\times \left( V(0) \exp^{-\sigma T_i} + \eta(i) + \exp^{-(\sigma + \psi )N(i) \overline{T}} \frac{M^2}{4} \right).  
\end{align*}
Now, using the fact that $N(i+1) = N(i) +1$, we obtain
\begin{align*}
& V(T_{i+1}) 
\leq \exp^{(\sigma + \psi )N(i+1) \overline{T}} \exp^{\psi \overline{T}} \exp^{-(\sigma + \psi ) \overline{T}} 
\\ &
\times \left( V(0) \exp^{-\sigma T_i} + \eta(i) + \exp^{-(\sigma + \psi )N(i) \overline{T}} \frac{M^2}{4} \right)
\\ & \leq 
\exp^{(\sigma + \psi )N(i+1) \overline{T}} \exp^{\psi  \overline{T}}  
\\ &
\times \left( V(0) \exp^{-\sigma T_i} + \eta(i) \exp^{-(\sigma + \psi ) \overline{T}} + \frac{M^2}{4} \right). 
\end{align*}
The proof is completed by showing that 
$$ \eta(i) \exp^{-(\sigma + \psi ) \overline{T}} + M^2/4 \leq \eta(i+1).  $$
Indeed, we note that
\begin{align*}
& \eta(i) \exp^{-(\sigma + \psi ) \overline{T}} + \frac{M^2}{4}  \\ & \leq 
\eta(i) \exp^{-\sigma  \overline{T}} \exp^{-\psi  \overline{T}} + \left(M+ \frac{M^2}{4}\right) \exp^{\psi \overline{T}} 
\\ & \leq 
\left(\sum_{k=1}^{i+1}
\exp^{-k\sigma \underline{T}} \right) \left(M + \frac{ M^2}{4} \right) \exp^{\psi 
\overline{T}}  
\\ & + \left(M+ \frac{M^2}{4}\right) \exp^{\psi \overline{T}} = \eta(i+1).   
\end{align*}
\end{proof}

\subsection*{Proof of Lemma \ref{v_space_vary}}
By differentiating $V_{1}$ along \eqref{twopdesB}, we obtain 
    \begin{equation}
    \label{proofGA11}
    \begin{aligned}
    \dot{V}_{1} & = \int_{0}^{Y}w(x)w_{t}(x)dx  \\
    & = \int_{0}^{Y} w(x) \left[-w(x)w_{x}(x)-\lambda(x) w_{xx}(x) \right.  \\ & 
    \left. -w_{xxxx}(x)+f(x)\right] dx. \nonumber
    \end{aligned}
    \end{equation}
Note that  
 $-3 \int_{0}^{Y} w(x)^2 w_{x}(x) dx = - w(Y)^3 +  w(0)^3$.    
 Using integration by part, we obtain
 \begin{align*}
 & - \int^Y_{0} w(x) w_{xxxx}(x) dx
 \\ & = - \left[ w(x) w_{xxx}(x) \right]^Y_{0} + \int^Y_0 w_x(x) w_{xxx}(x) dx 
 \\ &
 = - \left[ w(x) w_{xxx}(x) \right]^Y_{0} + \left[  w_x(x) w_{xx}(x) \right]^Y_{0} - \int^Y_0 w_{xx}(x)^2 dx. 
 \end{align*}
Using the boundary conditions $w_x(0)=w_x(Y)=0$, we obtain 
 \begin{align*}
  \int^Y_{0} w(x) w_{xxxx}(x) dx & =  {\left[ w(x) w_{xxx}(x) \right]}^Y_{0}  - \int^Y_0 w_{xx}(x)^2 dx
 \\ & 
 = -  u_1 w_{xxx}(0)  + \int^Y_0 w_{xx}(x)^2 dx.
 \end{align*}

Similarly, note that
\begin{align*}
-\int_{0}^{Y}\lambda(x) w(x)w_{xx}(x)dx  =&~ \int_{0}^{Y}\lambda(x) w_{x}(x)^2dx \\
&~+\int_{0}^{Y}\lambda'(x)w(x)w_x(x)dx.
\end{align*}
Using Young inequality, we obtain 
$$|\lambda'(x)w(x)w_x(x)|\leq \frac{1}{2}(\bar{\lambda}_{l}^{'2}w(x)^2+w_x(x)^2).$$
This allows us to conclude that 
\begin{align*}
-\int_{0}^{Y}\lambda(x) w(x)w_{xx}(x)dx &\leq \left(\bar{\lambda}_l+\frac{1}{2}\right)\int_{0}^{Y}w_x(x)^2dx \\
&+\bar{\lambda}_{l}^{'2}V_1.
\end{align*}

Finally, using Cauchy-Schwarz inequality, we obtain $\int_{0}^{Y} \hspace{-0.1cm} w(x)f(x)dx \leq \bar{f}\int_{0}^{Y} \hspace{-0.2cm} |w(x)|dx \leq C_{1}\sqrt{V_{1}}$. As a consequence, we have
\begin{align*}
\dot{V}_1 & \leq \bar{\lambda}_{l}^{'2}V_1+\left(\bar{\lambda}_{l}+\frac{1}{2}\right)\int_{0}^{Y}w_x(x)^2dx \nonumber \\
& - \int_{0}^{Y}w_{xx}(x)^2dx + C_1\sqrt{V_1} +\frac{u_1^3}{3}+u_1w_{xxx}(0). 
\end{align*}
Invoking Lemma \ref{halperin-pitt} with $\epsilon := 1/(\bar{\lambda}_{l}+\frac{1}{2})$, we find 
\begin{align*}
\left(\bar{\lambda}_{l}+\frac{1}{2}\right)\int_{0}^{Y} w_x(x)^2dx&- \int_{0}^{Y}w_{xx}(x)^2dx \\
&\leq (\theta_1-\bar{\lambda}_{l}^{'2})V_1,
\end{align*}
which proves inequality \eqref{v1_new}. We show inequality \eqref{v2_new} in a similar way. 
\hfill $\blacksquare$

\subsection*{Proof of Lemma \ref{lem_control}}
Consider first the case when $|\omega |\geq l(V,\hat{\theta})$. It implies that $\kappa^3+3\kappa \omega \leq V-3\sqrt[3]{V}l(V,\hat{\theta})\leq -3\hat{\theta}V$. On the other hand, $|\omega|<l(V,\hat{\theta})$ implies that
\begin{align*}
\kappa^3+ & 3\kappa \omega +3\hat{\theta}V\leq \kappa^3+3|\kappa |l(V,\hat{\theta})+3\hat{\theta}V \\
\leq&~ \big[ -\varepsilon^3 (\hat{\theta}+\delta)^3 +\varepsilon (\hat{\theta}+\delta)(1+3\hat{\theta})+3\hat{\theta}\big]V \\
\leq&~ \big[-\varepsilon^3 \hat{\theta}^3 - 3 \varepsilon (\varepsilon (\varepsilon \delta)-1)\hat{\theta}^2 \\
&~-3(\varepsilon^3\delta^2-(\varepsilon /3)-\varepsilon \delta - 1)\hat{\theta} 
- (\varepsilon \delta)^3 + \varepsilon \delta\big]V.
\end{align*}
Since $\hat{\theta}\geq 0$, we conclude that
$\kappa^3+3\kappa \omega +3\hat{\theta}V \leq 0.$
\hfill $\blacksquare$

\else 

\fi


\begin{thebibliography}{10}

\bibitem{intermittent_app1}
H. Josiah and S. Jacob.
\newblock {The Future of Sensing is Batteryless, Intermittent, and Awesome}.
\newblock {\em In Proceedings of the 2017 ACM Conference on Embedded Network Sensor Systems}.

\bibitem{intermittent_app5}
S. Lan, Z. Wang, J. Mamish, J. Hester and Q. Zhu. 
\newblock {AdaSens: Adaptive Environment Monitoring by Coordinating Intermittently-Powered Sensors}.
\newblock {\em In Proceedings of the 2022 Asia and South Pacific Design Automation Conference}, pp. 556-561.


\bibitem{intermittent_app4}
A. Mainwaring, D. Culler, J. Polastre, R. Szewczyk, and J. Anderson.
\newblock {Wireless sensor networks for habitat monitoring}.
\newblock {\em In Proceedings of the 2002 ACM international workshop on Wireless sensor networks and applications. Association for Computing Machinery, New York, NY, USA}, pp. 88–97.

\bibitem{ew}
E. Witrant, P. Di Marco, P. Park and C. Briat.
\newblock {Limitations
and performances of robust control over WSN: UFAD control
in intelligent buildings}.
\newblock {\em IMA Journal of Mathematical
Control and Information, vol. 27, no. 4, pp. 527–543}, 2010.

\bibitem{intermittent_app3}
L. Yanjun, W. Zhi and S. Yeqiong.
\newblock {Wireless Sensor Network Design for Wildfire Monitoring}.
\newblock {\em In Proceedings of the 2006 World Congress on Intelligent Control and Automation}, pp. 109-113.



\bibitem{kuramoto78}
Y.~Kuramoto.
\newblock Diffusion-induced chaos in reaction systems.
\newblock {\em Progress of Theoretical Physics Supplement}, 64:346--367, 1978.

\bibitem{sivashinsky80}
G.~Sivashinsky.
\newblock On flame propagation under conditions of stoichiometry.
\newblock {\em SIAM Journal on Applied Mathematics}, 39(1):67–82, 1980.

\bibitem{wildfire}
N. Larkin. 
\newblock {Global Solutions for the Kuramoto-Sivashinsky Equation Posed on Unbounded 3D Grooves}.
\newblock {\em Contemp. Math}, 2(4):293-30, 2021.

\bibitem{plasmaInstable1}
R.~LaQuey, S.~Mahajan, P.~Rutherford, and W.~Tang.
\newblock {Nonlinear saturation of the trapped-ion mode}.
\newblock {\em Physical Review Letters}, 33:391–394, 1975.


\bibitem{liu2001stability}
W-J.~Liu and M.~Krsti{\'c}.
\newblock Stability enhancement by boundary control in the
  {Kuramoto--Sivashinsky} equation.
\newblock {\em Nonlinear Analysis: Theory, Methods \& Applications},
  43(4):485--507, 2001.

\bibitem{kobayashi}
T.~Kobayashi.
\newblock Adaptive stabilization of the {Kuramoto-Sivashinsky} equation.
\newblock {\em Internat. J. Systems Sci.}, 33:3:175--180, 2002.

\bibitem{coron2015fredholm}
J-M. Coron and Q.~L{\"u}.
\newblock {Fredholm transform and local rapid stabilization for a
  Kuramoto--Sivashinsky equation}.
\newblock {\em Journal of Differential Equations}, 259(8):3683--3729, 2015.

\bibitem{HugoKS}
H.~Lhachemi.
\newblock Local output feedback stabilization of a nonlinear Kuramoto-Sivashinsky equation.
\newblock {\em IEEE Trans. on Automatic Control}, pages 1--6, 2023.

\bibitem{vs1}
D. Gajardo, A. Mercado, and J.~C. Mu{\~n}oz.
\newblock {Identification of the anti-diffusion coefficient for the linear
  Kuramoto-Sivashinsky equation}.
\newblock {\em Journal of Mathematical Analysis and Applications},
  495(2):124747, 2021.

\bibitem{vs2}
L. Baudouin, E. Cerpa, E. Cr{\'e}peau, and A. Mercado.
\newblock {Lipschitz stability in an inverse problem for the
  Kuramoto--Sivashinsky equation}.
\newblock {\em Applicable Analysis}, 92(10):2084--2102, 2013.


\bibitem{KS1}
M.~Maghenem, C.~Prieur, and E.~Witrant.
\newblock Boundary control of the {Kuramoto-Sivashinsky} equation under
  intermittent data availability.
\newblock {\em In Proceedings of the 2022 American Control Conference, 2227-2232, Atlanta, GA, USA}.

\bibitem{ACC23-KS}
M.C. Belhadjoudja, M.~Maghenem, E.~Witrant, and C.~Prieur.
\newblock Adaptive stabilization of the {Kuramoto-Sivashinsky} equation subject
  to intermittent sensing.
\newblock {\em In Proceedings of the 2023 American Control Conference, 1608-1613, San Diego, CA, USA}.

\bibitem{caoInh}
H.~Cao, D.~Lu, L.~Tian, and Y.~Cheng.
\newblock Boundary control of the {Kuramoto-Sivashinsky} equation with an
  external excitation.
\newblock {\em International Journal of Nonlinear Science}, 1(2):67--81, 2006.

\bibitem{ISS_PDE}
S. Dashkovskiy, A. Mironchenko. 
\newblock Input-to-state stability of infinite-dimensional control systems.
\newblock {\em Math. Control Signals Syst. 25, 1–35}, 2013.

\bibitem{network}
W. Zhang, M. S. Branicky, and S. M. Phillips. 
\newblock Stability of networked control systems. 
\newblock {\em IEEE control systems magazine}, 21(1), 84-99, 2001.

\bibitem{scanning}
M. A. Demetriou. 
\newblock Guidance of Mobile Actuator-Plus-Sensor Networks for Improved Control and Estimation of Distributed Parameter Systems.
\newblock {\em IEEE Trans. on Automatic Control, vol. 55, no. 7, pp. 1570-1584}, 2010.

\bibitem{strong}
D.~Gilbarg, N.~S.~Trudinger, D.~Gilbarg, and 
N.~S.~Trudinger.
\newblock {\em {Elliptic partial differential equations of second order}},
  volume 224.
\newblock Springer, 1977.

\bibitem{discontinuous1}
T. Liard, I. Balogoun, S. Marx, and F. Plestan.
\newblock Boundary sliding mode control of a system of linear hyperbolic
  equations: a Lyapunov approach.
\newblock {\em Automatica}, 135:109964, 2022.

\bibitem{discontinuous2}
J-M.~Wang, J-J.~Liu, B.~Ren, and J.~Chen.
\newblock Sliding mode control to stabilization of cascaded heat PDE--ODE
  systems subject to boundary control matched disturbance.
\newblock {\em Automatica}, 52:23--34, 2015.

\ifitsdraft

\else 
\bibitem{preprint}
M.~C. Belhadjoudja, M.~Maghenem, E.~Witrant, and C.~Prieur.
\newblock Adaptive Boundary Control of the {Kuramoto-Sivashinsky}
Equation Under Intermittent Sensing.
\newblock {\em Preprint, arXiv:2403.18055}, 2024.
\fi


\ifitsdraft

\bibitem{PITT}
D.~S.~Mitrinovic, J.~E.~ Pecaric, and A.~M.~Fink.
\newblock {\em Inequalities involving functions and their integrals and
  derivatives}.
\newblock {Kluwer Academic Publishers}, Dordrecht/Boston/London, 1991.
\fi 


\bibitem{collocation}
M. Uddin, S. Haq, and S. ul~Islam.
\newblock A mesh-free numerical method for solution of the family of
  {Kuramoto–Sivashinsky} equations.
\newblock {\em {Applied Mathematics and Computations. 212}, pages 458--469}, 2009.

\bibitem{kdv}
A. Balogh, and M. Krsti\'c.
\newblock Boundary control of the Korteweg-de Vries-Burgers equation: further results on stabilization and well-posedness, with numerical demonstration.
\newblock {\em IEEE Trans. on Automatic Control, vol. 45, no 9, p. 1739-1745}, 2000.

\bibitem{intermittent_app7}
P.~Parmananda.
\newblock Generalized synchronization of spatiotemporal chemical chaos.
\newblock {\em Phys. Rev. E}, 56:1595--1598, 1997.


\end{thebibliography}
\end{document}